\RequirePackage{fix-cm}
\documentclass{article}
\usepackage[T1]{fontenc}
\usepackage[utf8]{inputenc}
\usepackage[english]{babel}
\usepackage{fullpage}
\usepackage{authblk}
\usepackage{lmodern}
\usepackage{amssymb}
\usepackage{amsmath}
\usepackage{amsthm}
\usepackage{graphicx}
\usepackage{enumerate}
\usepackage{xspace}
\usepackage[ruled, linesnumbered]{algorithm2e}
\usepackage[colorlinks,bookmarksopen,bookmarksnumbered,citecolor=blue,linkcolor=magenta]{hyperref}

\newtheorem{theorem}{Theorem}
\newtheorem{lemma}{Lemma}
\newtheorem{corollary}{Corollary}
\theoremstyle{definition}
\newtheorem{definition}{Definition}

\newcommand{\PP}{\textrm{P}\xspace}
\newcommand{\NP}{\textrm{NP}\xspace}
\newcommand{\coNP}{\textrm{coNP}\xspace}

\newcommand{\etal}{et~al.\@\xspace}
\newcommand{\Rnn}{\mathbb{R}_{+}}
\newcommand{\Znn}{\mathbb{Z}_{+}}

\newcommand{\OPT}{\mathrm{OPT}}
\newcommand{\ILP}{{\textit{ILP}_{k,\alpha}(G)}}
\newcommand{\ILPn}{{\textit{ILP}_{k,\alpha}(G_n)}}
\newcommand{\LP}{{\textit{LP}_{k,\alpha}(G,\{C_v\}_{v \in \Gamma})}}
\newcommand{\LPUU}{{\textit{LPU}_{k,\alpha}(G)}}

\newcommand{\y}[1]{y^{\scriptscriptstyle(#1)}}
\newcommand{\X}[1]{X^{\scriptscriptstyle(#1)}}

\newcommand{\llIf}[2]{{\let\par\relax\lIf{#1}{#2}}}
\newcommand{\llElse}[1]{{\let\par\relax\lElse{#1}}}

\newcommand{\failure}{\textsc{failure}\xspace}
\newcommand{\alg}{\textsc{alg}\xspace}

\DeclareMathOperator*{\argmax}{arg\,max}

\begin{document}

\title{Improved Approximation Algorithms for Capacitated Fault-Tolerant $k$-Center%
\thanks{Partially supported by CAPES, CNPq (grants 308523/2012-1,
477203/2012-4, and 456792/2014-7), FAPESP (grants \mbox{2013/03447-6}
and \mbox{2014/14209-1}), and MaCLinC.}}

\author{Cristina G.\@ Fernandes}
\author{Samuel P.\@ de Paula}
\affil{Department of Computer Science, University of S\~ao Paulo, Brazil\\
\texttt{\{cris,samuelp\}@ime.usp.br}}
\author{Lehilton L.\@ C.\@ Pedrosa}
\affil{Institute of Computing, University of Campinas, Brazil\\
\texttt{lehilton@ic.unicamp.br}}

\date{}

\maketitle

\begin{abstract}
In the $k$-center problem, given a metric space~$V$ and a positive integer~$k$,
one wants to select $k$ elements (centers) of~$V$ and an assignment from~$V$ to
centers, minimizing the maximum distance between an element of~$V$ and its
assigned center. One of the most general variants is the \emph{capacitated
$\alpha$-fault-tolerant $k$-center}, where centers have a limit on the number of
assigned elements, and, if $\alpha$ centers fail, there is a reassignment
from~$V$ to non-faulty centers.
In this paper, we present a new approach to tackle fault tolerance, by selecting
and pre-opening a set of backup centers, then solving the obtained residual
instance.
For the $\{0,L\}$-capacitated case, we give approximations with factor $6$ for
the basic problem, and $7$ for the so called \emph{conservative} variant, when
only clients whose centers failed may be reassigned. Our algorithms improve on
the previously best known factors of $9$ and $17$, respectively. Moreover, we
consider the case with general capacities. Assuming $\alpha$ is constant, our
method leads to the first approximations for this case. We also derive
approximations for the capacitated fault-tolerant $k$-supplier problem.
\end{abstract}

\section{Introduction}

The \emph{$k$-center} is the minimax problem in which, given a metric space $V$
and a positive integer $k$, we want to choose a set of $k$ centers such that the
maximum distance from an element of $V$ to its closest center is minimum. More
precisely, the goal is to select $S \subseteq V$ with $|S| = k$ that minimizes
\[
\max_{u \in V} \min_{v \in S} d(u,v),
\]
where $d(u,v)$ is the distance between $u$ and $v$. The decision version of the
$k$-center appears in Garey and Johnson's list of NP-complete problems,
identified by MS9~\cite{GareyJ79}. It is well known that $k$-center has a
$2$-approximation which is best possible unless $\PP =
\NP$~\cite{FederG88,Gonzalez85,HochbaumS85,HochbaumS86,HsuN79}. The elements of
set~$S$ are usually referred to as \emph{centers}, and the elements of~$V$ as
\emph{clients}.

In a typical application of $k$-center, set $V$ represents the nodes of a
network, and one may want to install $k$ routers so that the network latency is
minimized. Other applications have additional constraints, so variants of the
$k$-center have been considered as well. For example, the number of nodes that a
router may serve might be limited.
In the \emph{capacitated $k$-center}, in addition to the set of selected
centers, we also want to obtain an assignment from the set of clients to centers
such that at most a number $L_u$ of clients are assigned to each center $u$. The
value~$L_u$ is called the \emph{capacity} of $u$. The first approximation for
this version of the problem is due to Bar-Ilan~\etal~\cite{BarilanKP93}, who
gave a $10$-approximation for the particular case of uniform capacities, where
there is a number $L$ such that $L_u = L$ for every $u$ in $V$. This was
improved by Khuller and Sussmann~\cite{KhullerS00}, who obtained a
$6$-approximation, and also considered the \emph{soft} capacitated case, in
which multiple centers may be opened at the same location, obtaining a
$5$-approximation, both results for uniform capacities.

Despite the progress in the approximation algorithms for related problems, such
as the metric facility location problem, the first constant approximation for
the (non-uniformly) capacitated $k$-center was obtained only in 2012, by
Cygan~\etal~\cite{CyganHK12}. Differently from algorithms for the uniform case,
the algorithm of Cygan~\etal is based on the relaxation of a linear programming
(LP) formulation. Since the natural formulation for the $k$-center has unbounded
integrality gap, a preprocessing is used, what allows considering only instances
whose LP has bounded gap. They also presented an $11$-approximation for the soft
capacitated case.
Later, An~\etal~\cite{AnBCGMS15} presented a cleaner rounding algorithm, and
obtained an improved approximation with factor~$9$ (while the previous
approximation had a large constant factor, not explicitly calculated).
Cygan~\etal~\cite{CyganK14} also presented an algorithm for a variant of the
problem with outliers.
As for negative results, it has been shown that the capacitated $k$-center has
no approximation with factor better than $3$ unless $\PP =
\NP$~\cite{CyganHK12}.

Another natural variant of the $k$-center comprises the possibility that centers
may fail during operation. This was first discussed by Krumke~\cite{Krumke95},
who considered the version in which clients must be connected to a given minimum
number of centers. In the \emph{fault-tolerant $k$-center}, for a number
$\alpha$, we consider the possibility that any subset of centers of size at most
$\alpha$ may fail. The objective is to minimize the maximum distance from a
client to its $\alpha + 1$ nearest centers. For the variant in which selected
centers do not need to be served, Krumke~\cite{Krumke95} gave a
$4$-approximation, later improved to a (best possible) $2$-approximation by
Chaudhuri~\etal~\cite{ChaudhuriGR98}, and Khuller~\etal~\cite{KhullerPS00}. For
the standard version, in which a client must be served even if a center is
installed at the client's location, there is a $3$-approximation by
Khuller~\etal~\cite{KhullerPS00}, who also gave a $2$-approximation for the
particular case of $\alpha \le 2$.

Chechik and Peleg~\cite{ChechikP15} considered a common generalization of the
capacitated $k$-center and the fault-tolerant $k$-center, where centers have
limited capacity and may fail during operation.  They defined only the uniformly
capacitated version, presenting a $9$-approximation. Also, they considered the
case in which, after failures, only clients that were assigned to faulty centers
may be reassigned. For this variant, called the conservative fault-tolerant
$k$-center, a $17$-approximation was obtained for the uniformly capacitated
case. For the special case in which $\alpha < L$, the so called \emph{large
capacities} case, they obtained a $13$-approximation.

\subsection{Our contributions and techniques}

We consider the \emph{capacitated $\alpha$-fault-tolerant $k$-center} problem.
Formally, an instance for this problem consists of a metric space $V$ with
corresponding distance function $d : V\times V \rightarrow \Rnn$, non-negative
integers $k$ and $\alpha$, with $\alpha < k$, and a non-negative integer~$L_v$
for each~$v$ in~$V$. A solution is a subset $S$ of $V$ with ${|S| = k}$, such
that, for each $F \subseteq S$ with~$|F| \le \alpha$, there exists an
assignment~$\phi_F : V \rightarrow S \setminus F$ with $|\phi_F^{-1}(v)| \le
L_v$ for each $v$ in $S\setminus F$. For a given~$F$, we denote by $\phi^*_F$ an
assignment $\phi_F$ with minimum $\max_{u\in V} d(u, \phi_F(u))$.
The problem's objective is to find a solution that minimizes
\[
\max_{u \in V, F \subseteq V : |F| \le \alpha} d(u,\phi^*_F(u)).
\]

We also consider the \emph{capacitated conservative $\alpha$-fault-tolerant
$k$-center}. In this variant, in addition to the set $S$, a solution comprises
an initial assignment~$\phi_0$. We require that an assignment~$\phi_F$ for a
failure scenario~$F$ differs from~$\phi_0$ only for vertices assigned by
$\phi_0$ to centers in~$F$.
Precisely, given $F\subseteq S$ with ${|F| \leq \alpha}$, we say that an
assignment $\phi_F$ is \emph{conservative} (with respect to $\phi_0$) if
$\phi_F(u) = \phi_0(u)$ for every~$u \in V$ with~$\phi_0(u) \notin F$. A
solution for the problem is a pair $(S, \phi_0)$ such that, for each $F\subseteq
S$ with~${|F| \leq \alpha}$, there exists a conservative assignment~$\phi_F$.
The objective function is defined analogously.

Our major technical contribution is a new strategy to deal with the
fault-tolerant and capacitated problems. Namely, we solve the considered
problems in two phases. In the first phase, we identify clusters of vertices
where an optimal solution must install a minimum of $\alpha$ centers. For each
cluster, we carefully select~$\alpha$ of its vertices, and pre-open them as
centers. These $\alpha$ centers will have enough \emph{backup} capacity so that,
in the case of failure events, the unused capacity of all pre-opened centers
will be sufficient to obtain a reassignment for all clients. While the $\alpha$
guessed centers of a cluster may not correspond to centers in an optimal
solution, we carefully select elements that are near to centers of an optimal
solution, so that our choice leads to an approximate solution.
In the second phase, we are left with a residual instance, where part of a
solution is already known. Depending on the problem, obtaining the remaining
centers of a solution may be reduced to the non-fault-tolerant variant.
Otherwise, we can make stronger assumptions over the input and the solution, so
that the task of obtaining a fault-tolerant solution is simplified.

A good feature of the presented approach is that it can be used in combination
with different methods and algorithms, and can be applied to different versions
of the problem. Indeed, we obtain approximations for both the conservative and
non-conservative variants of the capacitated fault-tolerant $k$-center.
Moreover, each of the obtained approximations uses novel and specific techniques
that are of particular interest. For the conservative variant, we present
elegant combinatorial algorithms that reduce the problem to the
non-fault-tolerant case. For the non-conservative variant, our algorithms are
based on the rounding of a new LP formulation for the problem. Interestingly, we
use the set of pre-opened centers to obtain a partial solution for the LP
variables with integral values. We hope that other problems can benefit from
similar techniques.

\subsection{Obtained approximations and paper organization}

The conservative variant is considered in
Sections~\ref{sec:conservative-aprox0L} and~\ref{sec:conservative-nonuniform}.
In Section~\ref{sec:conservative-aprox0L}, we present a $7$-approximation for
the $\{0,L\}$-capacitated conservative $\alpha$-fault-tolerant $k$-center. This
is the subset of the problem where the capacities are either~$0$ or~$L$, for
some~$L$. Notice that this generalizes the uniformly capacitated case, when all
capacities are equal to~$L$.
This result improves on the previously known factors of $17$ and $13$ by Chechik
and Peleg~\cite{ChechikP15}, that apply to particular cases with uniform
capacities, and uniform large capacities, respectively.
In Section~\ref{sec:conservative-nonuniform}, we study the case of general
capacities, and present a $(9 + 6\alpha)$-approximation when $\alpha$ is
constant. To the best of our knowledge, this is the first approximation for the
problem with arbitrary capacities.

For the non-conservative variant, our algorithms are based on the rounding of a
new LP formulation, and are described in Sections~\ref{sec:noncons}
and~\ref{sec:aprox0L}.
First we consider the case of arbitrary capacities in~Section~\ref{sec:noncons}.
We present the LP formulation, and give a $10$-approximation when $\alpha$ is
constant. Once again, this is the first approximation for the problem with
arbitrary capacities.
In Section~\ref{sec:aprox0L}, the rounding algorithm is adapted for the
$\{0,L\}$-capacitated fault-tolerant $k$-center, for which we obtain a
$6$-approximation with $\alpha$ being part of the input. This factor matches the
best known factor for the problem without fault
tolerance~\cite{AnBCGMS15,KhullerS00}, and improves on the best previously known
algorithm for the fault-tolerant version, which achieves factor~$9$ for the
uniformly capacitated case~\cite{ChechikP15}.

In Section~\ref{sec-supplier}, we apply our technique to the $k$-supplier
problem,
and in Section~\ref{sec-complex} we show some complexity results for related
decision problems for in case that $\alpha$ is part of the input.
A summary of the results is given in Table~\ref{tab:summary}.

\begin{table}[h]
\renewcommand{\arraystretch}{1.1}
\centering
\begin{tabular}{|@{\ }c@{\ }|@{\ }c@{\ }|@{\ }c@{\ }|@{\ }c@{\ }|@{\ }c@{\ }|}
\hline
\textbf{Version} & \textbf{Capacities} & \textbf{Value of $\alpha$}
& \textbf{Previous} & \textbf{This paper} \\
\hline
conservative     & uniform   & given in the input & $17$~\cite{ChechikP15} & $7$         \\
conservative     & arbitrary & fixed              & --                     & $9+6\alpha$ \\
non-conservative & uniform   & given in the input & $9$~\cite{ChechikP15}  & $6$         \\
non-conservative & arbitrary & fixed              & --                     & $10$        \\\hline
\end{tabular}
\caption{\label{tab:summary} Summary of the obtained approximation factors for
the $k$-center problem.}
\end{table}

\section{Preliminaries}
\label{sec:prel}

Let $G = (V, E)$ be an undirected and unweighted graph. We denote by $d_G$ the
metric induced by $G$, that is, for~$u$ and $v$ in $V$, let $d_G(u,v)$ be the
length of a shortest path between $u$ and~$v$ in $G$.
For given nonempty sets $A$, $B \subseteq V$, we define $d_G(A,B) = \min_{a\in
A, b\in B} d_G(a,b)$. Also, for $a \in V$, we may write $d_G(a, B)$ instead of
$d_G(\{ a \},B)$.

For an integer~$\ell$, we let ${N^\ell_G(u) = \{v \in V : d_G(v,u) \le \ell\}}$.
For a subset~$U \subseteq V$, let $N^\ell_G(U) = \bigcup_{u \in U} N^\ell_G(u)$.
We may omit the superscript $\ell$ when $\ell = 1$, and the subscript~$G$ when
the graph is clear from the context. For a directed graph~$G$, we define
$d_G(u,v)$ as the length of a shortest directed path from $u$ to~$v$ in~$G$, and
define $N^\ell_G(u)$ similarly. We also define the (power) graph $G^\ell =
(V,E^\ell)$, where $\{u,v\}\in E^{\ell}$ if $v\in N^{\ell}(u)\setminus\{u\}$.

\subsection{Reduction to the unweighted case}

As it is standard for the $k$-center problem, we will use the bottleneck
method~\cite{HochbaumS86}, so that we can consider the case in which the metric
space is induced by an unweighted undirected graph.
Suppose we have an algorithm that, given an unweighted graph, either produces
a distance-$r$ solution for the unweighted problem, or a certificate that no
distance-$1$ solution exists. We may then use this algorithm to obtain an
$r$-approximation for the general metric case.

Let $V$ be a metric space associated with distance function ${c : V \times V
\rightarrow \Rnn}$. For a certain number~$\tau$ in $\Rnn$, we consider the
\emph{threshold graph} defined as~${G_{\le \tau} = (V,E_{\le \tau})}$,
where~${E_{\le \tau} = \{\{u,v\} : c(u,v) \le \tau\}}$. Next we obtain a
sequence of values of $c(u,v)$ for $(u,v)$ in $V^2$, in increasing order. For
each~$\tau$ in this ordering, we obtain~$G_{\le \tau}$, and use the algorithm
for the unweighted case; we stop when the algorithm fails to provide a negative
certificate, and return the obtained solution. Notice that there must be a
distance-$1$ solution for~$G_{\le \OPT}$, where $\OPT$ denotes the optimum value
for the problem. Since $\OPT$ is in the considered ordering for~$\tau$, the
algorithm always stops, and returns a solution for some~$\tau \le \OPT$, so we
obtain a solution for the original problem of cost at most $r \cdot \tau \le r
\cdot \OPT$. Hence, from now on, we assume that an unweighted graph $G = (V, E)$ is
given, and that the goal is to either obtain a certificate that no distance-$1$
solution exists, or return a \mbox{distance-$r$} solution for some constant~$r$.

\subsection{Preprocessing and reduction to the connected case}
\label{sec:preprocessing}

We also may assume without loss of generality that $G$ is
connected~\cite{ChechikP15,CyganHK12,KhullerS00}. If this is not the case, we
may proceed as follows.
Suppose there is an algorithm that, given a connected graph~$\tilde{G}$, and an
integer $\tilde{k}$, produces a distance-$r$ solution with~$\tilde{k}$ vertices,
or gives a certificate that no distance-$1$ solution with $\tilde{k}$ vertices
exists. Now, consider a given arbitrary unweighted graph $G$, and a given
integer~$k$.
We decompose $G$ into its connected components, say $G_1, \dots, G_t$. For each
connected component~$G_i$, with $1 \le i \le t$, we run the algorithm for each
$\tilde{k} = \alpha+1, \dots, k$, and find the minimum value~$k_i$, if any, for
which the algorithm obtains a distance-$r$ solution. As the failure set is
arbitrary, in the worst case all faulty centers might be in the same component.
If, for some~$G_i$, there is no distance-$1$ solution with~$k$ centers
or if $k_1 + \dots + k_t > k$, then clearly there is no distance-$1$ solution
for~$G$ with~$k$ centers; otherwise, conjoining the solutions obtained for each
component leads to a distance-$r$ solution for $G$ with no more than $k$
centers, and this solution is tolerant to the failure of~$\alpha$ centers. From
now on, we will assume that $G$ is connected.

\section{$\{0,L\}$-Capacitated conservative fault-tolerant $k$-center}
\label{sec:conservative-aprox0L}

After the occurrence of a failure, a distance-$1$ conservative solution has to
reassign each unserved client to an open center in its vicinity with available
capacity. This requires some kind of ``local available center capacity'', to be
used as backup. The next definition describes a set of vertices that are nice
candidates to be open as backup centers. This set can be partitioned into
clusters of at most $\alpha$ vertices, with the clusters sufficiently apart from
each other.  The idea is that failures in the vicinity of one of these clusters
do not affect centers in the other clusters.  More precisely, the vicinities of
different clusters do not intersect, therefore, in a distance-$1$ conservative
solution, any client that is assigned to a center in a certain cluster cannot be
reassigned to a center in the vicinity of any of the other clusters.

\begin{definition}\label{alpha-independent}
Consider a graph $G=(V,E)$ and non-negative integers~$\alpha$ and $\ell$. A set
$W$ of vertices of $G$ is~\emph{$(\alpha,\ell)$-independent} if it can be
partitioned into sets $C_1, \dots, C_t$, such that $|C_i| \le \alpha$ for $1 \le
i \le t$, and $d(C_i, C_j) > \ell$ for $1 \le i < j \le t$.
\end{definition}

In what follows, we denote by~$(G,k,L,\alpha)$ an instance of the
capacitated conservative $\alpha$-fault-tolerant $k$-center as obtained by
Section~\ref{sec:prel}. We say that~$(G,k,L,\alpha)$ is feasible if there exists
a distance-$1$ solution for it.

\begin{lemma}\label{lem:alpha-independent}
Let~$(G,k,L,\alpha)$ be a feasible instance for the capacitated conservative
$\alpha$-fault-tolerant $k$-center, and let $(S^*, \phi^*_0)$ be a corresponding
\mbox{distance-$1$} solution. If \mbox{$W \subseteq S^*$} is an
$(\alpha,4)$-independent set in $G$, then ${(G,k-|W|,L')}$ is feasible for the
capacitated $k$-center, where \mbox{$L'_u = 0$} for~$u$ in $W$, and \mbox{$L'_u
= L_u$} otherwise.
\end{lemma}

\begin{proof}
  Since $W$ is~$(\alpha,4)$-independent, there must be a partition
  $C_1,\ldots,C_t$ of $W$ such that $d(C_i, C_j) > 4$ for any pair $i, j$,
  with $1 \le i < j \le t$. Also, each part~$C_i$ has at most~$\alpha$
  vertices, and thus there exists a conservative assignment $\phi^*_{C_i}$
  with \mbox{$(\phi^*_{C_i})^{-1}(C_i) = \emptyset$}. Therefore,
  $\phi^*_{C_i}$ is a distance-$1$ solution for the~$(G,k-|C_i|,L^i)$ instance
  of the capacitated $k$-center problem, where \mbox{$L^i_u = 0$} for~$u$ in
  $C_i$, and~\mbox{$L^i_u = L_u$} otherwise. Moreover, as $\phi^*_0$ is
  conservative, $\phi^*_{C_i}$ differs from $\phi^*_0$ only
  in~$(\phi^*_0)^{-1}(C_i)$.
  So, if a center $u$ in $S^*$ is such that $(\phi^*_0)^{-1}(u) \neq
  (\phi^*_{C_i})^{-1}(u)$, then $u \in N^2(C_i)$. As $W$
  is~$(\alpha,4)$-independent, $N^2(C_i) \cap N^2(C_j) = \emptyset$ for every
  $j \in [t] \setminus \{i\}$. Let $\psi$ be an assignment such that, for each
  client $v$,
  \begin{align*}
    \psi(v) = \begin{cases}
                  \phi^*_{C_i}(v)
                    &\mbox{ $\phi^*_0(v) \in C_i$ for some $i$ in $[t]$,} \\
                  \phi^*_0(v)
                    &\mbox{ otherwise.}
              \end{cases}
  \end{align*}
  Therefore, set $\psi^{-1}(u)$ is empty if $u \in W$; is
  $(\phi^*_{C_i})^{-1}(u)$ if there exists~$i \in [t]$ such that $u \in
  N^2(C_i) \setminus C_i$; and is $(\phi^*_0)^{-1}(u)$ otherwise. This means
  that, for~$L'$ as in the statement of the lemma, $|\psi^{-1}(u)| \leq L'_u$
  for every $u$, and so~$(S^*, \psi)$ is a solution for the $(G,k-|W|,L')$
  instance of the capacitated $k$-center problem.
\end{proof}

A set of vertices $A \subseteq V$ is \emph{$7$-independent} in~$G$ if every pair
of vertices in~$A$ is at distance at least~$7$ in~$G$. This definition was also
used by Chechik and Peleg~\cite{ChechikP15} and, as we will show, such a set is
useful to obtain an~$(\alpha,4)$-independent set in~$G$.

\begin{lemma}\label{lem:7independent}
Let $A$ be a $7$-independent set in $G$, for each~$a$ in $A$, let~$B(a)$ be any
set of $\alpha$ vertices in~$N(a)$, and let $B = \cup_{a \in A} B(a)$.
If~$(G,k,L,\alpha)$ is feasible for the capacitated conservative
$\alpha$-fault-tolerant $k$-center, then $(G,k - |B|,L')$ is feasible for the
capacitated $k$-center, where \mbox{$L'_u = 0$} for~$u$ in $B$, and \mbox{$L'_u
= L_u$} otherwise.
\end{lemma}

\begin{proof}
  Let~$(S^*,\phi_0^*)$ be a solution for~$(G,k,L,\alpha)$. For each $a \in A$,
  there must be at least $\alpha$ centers in $S^* \cap N(a)$. Let $W(a)$ be
  the union of $S^* \cap B(a)$ and other $\alpha - |S^* \cap B(a)|$ centers in
  $S^* \cap N(a)$. Let $W = \cup_{a \in A} W(a)$. Since $A$ is
  $7$-independent, $N^3(a)$ and $N^3(b)$ are disjoint for any two $a$ and $b$
  in $A$, and so~${N^2(W(a)) \cap N^2(W(b)) = \emptyset}$. Thus, $W$ is
  $(\alpha,4)$-independent.

  Now let $L''$ be such that $L''_u = 0$ if ${u \notin S^*}$, and $L''_u =
  L_u$ otherwise. Observe that the instance ${(G,k, L'', \alpha)}$ is feasible
  (as we only set to zero the capacities of non-centers). By
  Lemma~\ref{lem:alpha-independent}, the instance ${(G,k - |W|,L''')}$ is
  feasible, where ${L'''_u = 0}$ if ${u \in W}$, and ${L'''_u = L''_u}$ otherwise.
  Notice that $L'_u \ge L'''_u$ for every $u$, and $|B| = |W|$. Therefore,
  since ${(G,k - |W|,L''')}$ is feasible, so is ${(G,k - |B|,L')}$.
\end{proof}

Now we present a $7$-approximation for the $\{0,L\}$-capacitated conservative
$\alpha$-fault-tolerant $k$-center. For this case, rather than using a capacity
function, it is convenient to consider the subset of vertices with capacity~$L$,
that is denoted by~$V^L$. We denote by $(G, k, V^L, \alpha)$  and by $(G, k,
V^L)$ instances of the fault-tolerant and non-fault-tolerant versions.
The steps are detailed in Algorithm~\ref{alg:conservative-0L}, where \alg
denotes an approximation algorithm for the $\{0,L\}$-capacitated $k$-center.

\begin{algorithm}[H]
    \caption{$\{0,L\}$-capacitated conservative $\alpha$-fault-tolerant $k$-center.}
    \SetKwInOut{Input}{Input}
    \SetAlgoNoLine
    \DontPrintSemicolon
    \label{alg:conservative-0L}

    \BlankLine
    \Input{connected graph $G$, $k$, $V^L$, and $\alpha$}
    \BlankLine

    $A \gets$ a maximal $7$-independent vertex set in $G$\;
    \ForEach{$a \in A$}
    {
        $B(a) \gets \alpha$ vertices chosen arbitrarily in~$N(a) \cap V^L$\;
    }
    $B \gets \cup_{a\in A} B(a)$ \;
    \eIf{$\alg(G,k-|B|,V^L \setminus B)$ returns \failure}
    {
        \Return \failure\label{lin:failure}
    }
    {
        Let $(S,\phi)$ be the solution returned by $\alg(G,k-|B|,V^L \setminus B)$ \;
        \Return $(S \cup B,\phi)$\label{alg:returnsolution} \;
    }
\end{algorithm}

\begin{theorem}\label{thm:conservative0L}
If \alg is a $\beta$-approximation for the $\{0,L\}$-capacitated $k$-center,
then Algorithm~\ref{alg:conservative-0L} is a $\max\{7,\beta\}$-approximation
for the $\{0,L\}$-capacitated conservative $\alpha$-fault-tolerant $k$-center.
\end{theorem}

\begin{proof}
  Consider an instance $(G, k, V^L, \alpha)$ of the $\{0,L\}$-capacitated
  conservative $\alpha$-fault-tolerant $k$-center problem, with $G=(V,E)$.
  Let~$A$, $B(a)$ for $a$ in $A$, and~$B$ be as defined in
  Algorithm~\ref{alg:conservative-0L} with $(G, k, V^L, \alpha)$ as input.
  Assume that $(G, k, V^L, \alpha)$ is feasible. Since $A$ is $7$-independent,
  by Lemma~\ref{lem:7independent}, the instance $(G, k-|B|, V^L \setminus B)$,
  where we set to zero the capacities of all vertices in~$B$, is also feasible
  for the $\{0,L\}$-capacitated $k$-center problem. This means that, if
  Algorithm~\ref{alg:conservative-0L} executes Line~\ref{lin:failure}, then
  the given instance is indeed infeasible. On the other hand, if \textsc{alg}
  returns a solution~$(S,\phi)$, then, since $|S| \leq k - |B|$, the size of
  $S \cup B$ is at most~$k$, and $\phi$ is a valid initial center assignment.
  Moreover, $\phi$ is such that: (1)~each vertex~$u$ is at distance at
  most~$\beta$ from~$\phi(u)$; and (2)~no vertex is assigned to~$B$.

  Let $F \subseteq S \cup B$ with $|F|=\alpha$ be a failure scenario. We
  describe a conservative center reassignment for~$(S \cup B, \phi)$. We only
  need to reassign vertices initially assigned to centers in $F \setminus B$
  (as no vertex was assigned to a vertex in~$B$). Thus, at most $ L |F
  \setminus B|$ vertices need to be reassigned. For each such vertex~$u$, we
  can choose $a \in A$ at distance at most~$6$ from~$u$ (as~$A$ is maximal),
  and let $\tilde\phi(u) = a$.
  Then, for each $a \in A$, and for each $u$ with $\tilde\phi(u) = a$,
  reassign $u$ to some non-full center of $B(a) \setminus F$. Notice
  that~$B(a)\setminus F$ can absorb all reassigned vertices. Indeed, the
  available capacity of $B(a)\setminus F$ before the failure event is $L
  |B(a)\setminus F| = L|F \setminus B(a)| \ge L|F \setminus B|$, where we used
  $|B(a)| = |F| = \alpha$.
  Since for a reassigned  vertex $u$, $d(u, \tilde\phi(u)) \le 6$, and $u$ is
  reassigned to some center $ v \in N(\tilde\phi(u))$, the distance between
  $u$ and $v$ is at most~7. Also, if a vertex $u$ was not reassigned, then the
  distance to its center is at most $\beta$.
\end{proof}

Now, using the $6$-approximation for the $\{0,L\}$-capacitated $k$-center by
An~\etal~\cite{AnBCGMS15}, we obtain the following.

\begin{corollary}\label{cor:conservative0L}
Algorithm~\ref{alg:conservative-0L} using the algorithm by
An~\etal~\cite{AnBCGMS15} for the $\{0,L\}$-capacitated $k$-center is a
$7$-approximation for the $\{0,L\}$-capacitated conservative
$\alpha$-fault-tolerant $k$-center.
\end{corollary}

\section{Capacitated conservative fault-tolerant $k$-center}
\label{sec:conservative-nonuniform}

In this section, we consider the capacitated conservative
$\alpha$-fault-tolerant $k$-center. Recall that this is the case in which
capacities may be arbitrary. An instance for this problem is denoted by~${(G, k,
L, \alpha)}$ for some $G=(V,E)$ and \mbox{$L : V \rightarrow \Znn$}. Under the
assumption that~$\alpha$ is bounded by a constant, we present the first
approximation for the problem.

In the~$\{0,L\}$-capacitated case, each vertex assigned to a faulty center could
be reassigned to a non-faulty center in $B(a)$, for an arbitrary nearby element
$a$ of a $7$-independent set~$A$. Each $B(a)$ could absorb \emph{all} reassigned
vertices.
With arbitrary capacities, the set $B$ of pre-opened centers must be obtained
much more carefully, as the capacities of non-zero-capacitated vertices are not
necessarily all the same.
Once the set $B$ of backup centers is selected, one needs to ensure that the
residual instance for the capacitated $k$-center problem is feasible. In
Section~\ref{sec:conservative-aprox0L}, an $(\alpha,4)$-independent set is
obtained from $A$, and Lemma~\ref{lem:alpha-independent} is used. This lemma is
valid for arbitrary capacities, and so it is useful here as well. To obtain an
$(\alpha,4)$-independent set from $B$, we make sure that $B$ can be partitioned
in such a way that any two parts are at least at distance $7$. This is done by
Algorithm~\ref{alg:conservative}, where \alg denotes an approximation for the
capacitated $k$-center problem.

\begin{algorithm}[H]
    \caption{capacitated conservative $\alpha$-fault-tolerant $k$-center, fixed~$\alpha$.}
    \SetAlgoNoLine
    \SetKwInOut{Input}{Input}
    \DontPrintSemicolon
    \label{alg:conservative}

    \BlankLine
    \Input{connected graph $G=(V,E)$, $k$, and $L : V \rightarrow \Znn$}
    \BlankLine

    \ForEach{$u \in V$}
    {
      \lIf{$L_u > |V|$}{$L_u \gets |V|$}\label{lin:cutcap}
    }
    $B \gets \emptyset$\label{lin:ini-def-B}\;
    \While{there is a set $U \subseteq V$ with~$|U| \le \alpha$ and $L(U) > L(B\cap N^6(U))$\label{lin:test}}
    {
        $B\gets (B\setminus N^6(U))\cup U$ \label{alg:augment-B}\;
    }\label{lin:fim-def-B}
    \ForEach{$u \in V$}
    {
       \llIf{$u \in B$}{$L'_u \gets 0$} \llElse{$L'_u \gets L_u$}
    }
    \eIf{{\sc alg}$(G,k-|B|,L')$ returns \failure}
    {
        \Return \failure\label{lin:failure2}
    }
    {
      Let $(S,\phi)$ be the solution returned by \sc{alg}$(G,k-|B|,L')$\;
      \Return $(S \cup B,\phi)$
    }
\end{algorithm}

Algorithm~\ref{alg:conservative} is polynomial in the size of $G$, $k$, and~$L$.
The test in Line~\ref{lin:test} is equivalent to finding a set $U \subseteq V$
with~$|U| = \alpha$ that minimizes ${L(B \cap N^6(U)) - L(U)}$ (note that this
is a particular case of minimizing a submodular function with cardinality
constraint). If, for an arbitrary $\alpha$, there were a polynomial-time
algorithm for finding such a set~$U$, then Algorithm~\ref{alg:conservative}
would be polynomial also in $\alpha$. Unfortunately, as we show in
Section~\ref{sec-complex}, such algorithm only exists if $\PP = \coNP$.
When $\alpha$ is fixed, we may enumerate the sets~$U$ in polynomial time. In the
following, we show that Algorithm~\ref{alg:conservative} is an approximation
algorithm for the capacitated conservative fault-tolerant $k$-center assuming
that $\alpha$ is fixed.

Next lemma is the analogous of Lemma~\ref{lem:7independent}, but applies to the
case with general capacities.

\begin{lemma}\label{lem:preopenBgeral}
Let $B$ be the set of vertices obtained by Algorithm~\ref{alg:conservative}
after the execution of Lines~\ref{lin:ini-def-B}-\ref{lin:fim-def-B}. If the
instance ${(G,k,L,\alpha)}$ is feasible for the capacitated conservative
$\alpha$-fault-tolerant $k$-center, then the instance ${(G, k - |B|, L')}$ is
feasible for the capacitated $k$-center, where~${L'_u = 0}$ for~$u$ in $B$, and
${L'_u = L_u}$ otherwise.
\end{lemma}

\begin{proof}
  Recall Definition~\ref{alpha-independent}: a set of vertices
  is~$(\alpha,\ell)$-independent if it can be partitioned into sets $C_1,
  \dots, C_t$, such that $|C_i| \le \alpha$ for $1 \le i \le t$, and~$d(C_i,
  C_j) > \ell$ for $1 \le i < j \le t$. Let us argue that $B$ is
  $(\alpha,6)$-independent.

  Let $t$ be the number of components of $G^6[B]$ and take each $C_i$ to be
  the vertex set of one of the components of $G^6[B]$. Let us argue that
  $|C_i| \leq \alpha$ for every $i$ with $1 \leq i \leq t$. Suppose, for a
  contradiction, that $|C_i| > \alpha$ for some $i$ and let $U'$ be the
  vertices in $C_i$ that were inserted in~$B$ in the last iteration of
  Line~\ref{alg:augment-B} that affected $C_i$. Clearly $|U'| \leq \alpha$.
  Since $C_i$ corresponds to a connected component in $G^6[B]$ and $|C_i| >
  \alpha$, there must be a vertex in $C_i \setminus U'$ in~$N^6(U') \cup B$ at
  this execution of Line~\ref{alg:augment-B}, but then this vertex would have
  been removed from~$B$, a contradiction. So~$B$ is indeed
  $(\alpha,6)$-independent.

  Now, for each $i$ with $1 \le i \le t$, choose an arbitrary element $a_i$ in
  $C_i$. (Note that the set $A = \{a_1, \dots, a_t\}$ is $7$-independent
  in~$G$.)
  Consider a solution~$(S^*,\phi_0^*)$ for~$(G,k,L,\alpha)$ and observe that,
  for each $i$ with $1 \le i \le t$, there must be at least $\alpha+1$ centers
  in ${S^* \cap N(a_i)}$.  So let $W_i$ be the union of $C_i \cap S^*$ and
  other~$|C_i \setminus S^*|$ centers in~$S^* \cap N(a_i)$.  The set $W_i$ is
  well defined, as $|C_i| \le \alpha < |S^* \cap N(a_i)|$. Moreover, $|W_i| =
  |C_i|$ and $W_i \subseteq N(C_i)$.

  Let $W = \cup_{i = 1}^t W_i$ and note that $|W| = |B|$. For each
  pair~$i$,~$j$ with~${1 \le i < j \le t}$ we have that $d(W_i,W_j) > 4$,
  because $B$ is $(\alpha, 6)$-independent and thus $d(C_i, C_j) > 6$. Hence,
  $W$ is $(\alpha,4)$-independent.

  Let $L''$ be such that $L''_u = 0$ if ${u \notin S^*}$, and $L''_u = L_u$
  otherwise. Observe that the instance ${(G,k, L'', \alpha)}$ is feasible (as
  we only set to zero the capacities of non-centers). By
  Lemma~\ref{lem:alpha-independent}, the instance ${(G,k - |W|,L''')}$ is
  feasible, where ${L'''_u = 0}$ if $u \in W$, and ${L'''_u = L''_u}$
  otherwise.
  Notice that $L'_u \ge L'''_u$ for every $u$, and $|B| = |W|$. Therefore,
  since ${(G,k - |W|,L''')}$ is feasible, so is ${(G,k - |B|,L')}$.
\end{proof}

\begin{theorem}\label{thm:genconservative}
If \alg is a $\beta$-approximation for the capacitated $k$-center, then
Algorithm~\ref{alg:conservative} is a~$(\beta+6 \alpha)$-approximation for the
capacitated conservative $\alpha$-fault-tolerant $k$-center with fixed~$\alpha$.
\end{theorem}

\begin{proof}
  Let $(G,k,L, \alpha)$ be an instance of the capacitated conservative
  $\alpha$-fault-tolerant $k$-center. No center can have more than $|V|$
  clients assigned to it, so Line~\ref{lin:cutcap} does not affect a solution.

  Since $\alpha$ is fixed, each execution of Line~\ref{lin:test} takes time
  polynomial in~$|V|$.
  Also, each execution of Line~\ref{alg:augment-B} increases the value of
  $L(B)$ by at least one. But $L(B)$ is an integer, starts from 0, and is at
  most~$|V|^2$, because each vertex capacity is at most $|V|$ after executing
  Line~\ref{lin:cutcap}.
  Thus, the number of iterations is quadratic on $|V|$, and each one takes
  time polynomial in $|V|$. Finally, as \alg is a polynomial-time algorithm,
  we conclude that Algorithm~\ref{alg:conservative} is polynomial.

  By Lemma~\ref{lem:preopenBgeral}, we know that, if \alg returns \failure\ in
  Line~\ref{lin:failure2}, then the instance $(G,k,L,\alpha)$ is infeasible
  for the capacitated conservative $\alpha$-fault-tolerant $k$-center.
  On the other hand, if \alg returns a solution~$(S,\phi)$, then~$(S \cup
  B,\phi)$ is a valid set of centers and initial attribution for our problem,
  and is such that each vertex~$u$ is at distance at most~$\beta$
  from~$\phi(u)$. To complete our proof, we argue next that, for each failure
  scenario, each client~$u$ of a faulty center can be reassigned to a center
  at distance at most~$\beta+6\alpha$ from~$u$, and no center has its capacity
  exceeded by the reassignment.

  Consider a failure scenario~$F\subseteq V$ with $|F| = \alpha$. We define
  next a flow network $(H,c,s,t)$, with source $s$ and sink $t$, in which a
  maximum flow from~$s$ to~$t$ provides a valid distance-$(\beta{+}6\alpha)$
  reassignment for the clients of centers in~$F$ (see
  Figure~\ref{fig:flow-conservative}). Network graph $H = (V_H, E_H)$ (see
  figure below)  is such that the set $V_H$ of vertices is comprised of
  \begin{itemize}
    \item a copy of each~$y$ in $\phi^{-1}(F)$,
    \item a copy of each~$v$ in $F$,
    \item a copy of each~$u$ in $B \setminus F$,
    \item a second copy of each~$w$ in $B \cap F$, denoted by $\bar{w}$;
  \end{itemize}
  and, the set $E_H$ of arcs is comprised of
  \begin{itemize}
    \item for each~$y$ in $\phi^{-1}(F)$, an arc $(s,y)$ with capacity $c(s,y)=\infty$,
    \item for each~$v$ in $F$ and each~$y$ in $\phi^{-1}(v)$, an arc~$(y,v)$ with $c(y,v) = 1$,
    \item for each~$v$ in $F$ and each~$u$ in $B \cap N^6_G(v)$, an arc $(v,u)$ with $c(v,u) = \infty$,
    \item for each~$u$ in $B \setminus F$, an arc~$(u,t)$ with capacity $c(u,t) = L_u$, and,
    \item for each~$w$ in $B \cap F$, a (reversed) arc~$(\bar{w},w)$ with $c(\bar{w}, w) = \infty$.
  \end{itemize}

  \begin{figure}[h]
    \centering
    \scalebox{1.2}{\def\svgwidth{5.4cm}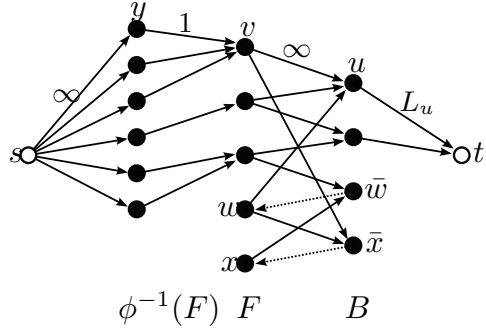}
    \caption{The flow network defined in terms of~$L$, $B$, $\phi$,
    and~$F$.\label{fig:flow-conservative}}
  \end{figure}

  Let~$C$ be the arcs of a minimum capacity $s$-$t$ cut in~$H$. We claim
  that~$c(C) = |\phi^{-1}(F)|$. Let $\delta^+(X)$ denote the set of arcs
  $(u,v)$ with $u \in X$, and $v \notin X$.
  Notice that $c(\delta^+(\phi^{-1}(F))) = 1 \cdot |\delta^+(\phi^{-1}(F))| =
  |\phi^{-1}(F)|$, so $c(C) \leq |\phi^{-1}(F)|$. Since only arcs from
  $\phi^{-1}(F)$ to~$F$ and from~$B$ to~$t$ have finite capacities, they are
  the only ones that can be in~$C$. Thus, there must be a set $U \subseteq F$
  such that
  \[
    C = \{ (y,v) : y \in \phi^{-1} (F\setminus U) \mbox{\ and } v = \phi(y)
    \} \;\; \cup \;\; ((N_H(U)\setminus F ) \times \{t\}),
  \]
  where the elements in $F$ refer to the first copy of each such element
  in~$V_H$.

  Let $Q =  N_G^6(U) \cap B \cap F$. We claim that $Q \subseteq U$. Let $v \in
  U$ and ${\bar{w} \in N_G^6(v) \cap B \cap F}$. Notice that arcs $(v,\bar{w})$
  and $(\bar{w}, w)$ have infinite capacities, and thus neither can be in the
  cut~$C$. It follows that $w \in U$, and indeed $Q \subseteq U$.

  Since there is no arc leaving $Q$ and reaching $t$, the previous equation
  allows us to express the capacity of~$C$ as
  \[
  c(C) = |\phi^{-1}(F\setminus U)| + L(N_H(U)\setminus Q).
  \]
  Notice that~$N_H(U) = (N_H(U)\setminus Q) \cup Q$. Hence, from the loop
  starting at Line~\ref{lin:test} of Algorithm~\ref{alg:conservative}, we have
  that~$L(U) \leq L(B\cap N^6_G(U)) = L(N_H(U)) = {L(N_H(U)\setminus Q) +
  L(Q)}$, and thus
  \begin{align*}
    c(C)&=    |\phi^{-1}(F\setminus U)| + L(N_H(U)\setminus Q) \\
        &\geq |\phi^{-1}(F\setminus U)| + L(U) - L(Q)\\
        &=    |\phi^{-1}(F\setminus U)| + |\phi^{-1}(U)| + (L(U) - |\phi^{-1}(U)|)   - L(Q)\\
        &\geq |\phi^{-1}(F\setminus U)| + |\phi^{-1}(U)| + L(Q)               - L(Q)\\
        &=    |\phi^{-1}(F\setminus U)| + |\phi^{-1}(U)| \ = \ |\phi^{-1}(F)|,
  \end{align*}
  where the second inequality comes from the fact that $\phi$ does not assign
  any vertex to~$Q$, and $Q \subseteq U$.

  So the value of a maximum integer flow on $H$ is exactly $|\phi^{-1}(F)|$,
  and thus every arc from $\phi^{-1}(F)$ to~$F$ has flow exactly~1.
  It is straightforward to obtain an assignment $\psi: \phi^{-1}(F)
  \rightarrow B \setminus F$. For each vertex~$y$ in $\phi^{-1}(F)$, let
  $\psi(y) = u$, where $u$ is the center in~$B\setminus F$ that receives the
  unit of flow going through $y$ (for example, in the previous figure, a unit
  of flow could traverse a path of vertices $s$, $y$, $v$, $\bar{x}$, $x$,
  $\bar{w}$, $w$, $u$, $t$, and so we set $\psi(y) = u$).

  Since each vertex $u$ in $B\setminus F$ can receive at most $L_u$ units of
  flow, clearly~$\psi$ respects the capacities. Moreover, since there are at
  most $\alpha -1$ elements in $B \cap F$, each unit of flow leaving a vertex
  $v$ in~$F$ can traverse at most $\alpha - 1$ reverse edges in $H$ (without
  creating a circle), so it can traverse at most $\alpha$ arcs from $F$ before
  reaching an element $u$ in $B \setminus F$. Therefore, $d_G(v, u) \le 6
  \alpha$.

  It follows that for every $y \in \phi^{-1}(F)$
  \[
  d_G(y,\psi(y)) \ \leq \ d_G(y,\phi(y)) + d_G(\phi(y),\psi(y)) \ \leq \ \beta
  + 6\alpha.
  \]

  Now we can define a conservative reassignment~$\phi_F$:
  \begin{align*}
    \phi_F(v) = \begin{cases}
                    \phi(v)    &\mbox{ if } \phi(v) \notin F, \\
                    \psi(v)    &\mbox{ otherwise.}
                \end{cases}
  \end{align*}

  We argue that~$\phi_F$ is a valid distance-$(\beta + 6\alpha)$ conservative
  reassignment.  Let~$u$ be a center opened by the algorithm (that is, $u \in
  S \cup B$). If $u \in S$, then $\phi_F^{-1}(u) = \phi^{-1}(u)$. If $u \in
  B$, then $\phi_F^{-1}(u) = \psi^{-1}(u)$. Also, both~$\phi$ and~$\psi$ do
  not exceed the capacities of the centers to which they assign unserved
  clients.  Finally, consider a vertex $y$ in~$V$.  If $\phi(y) \notin F$,
  then ${d_G(y,\phi_F(y)) = d_G(y,\phi(y)) \leq \beta}$. If~$\phi(y) \in F$,
  then $d_G(y,\phi_F(y)) = d_G(y,\psi(y)) \leq \beta + 6\alpha$.
\end{proof}

Using the best known approximation for the capacitated $k$-center, we obtain the
following.

\begin{corollary}
Algorithm~\ref{alg:conservative} using the algorithm by
An~\etal~\cite{AnBCGMS15} for the capacitated $k$-center is a $(9 +
6\alpha)$-approximation for the capacitated conservative $\alpha$-fault-tolerant
$k$-center with fixed~$\alpha$.
\end{corollary}

\section{Capacitated fault-tolerant $k$-center}
\label{sec:noncons}

\subsection{An initial LP formulation}
\label{subsec-iniLP}

Recall that we are given an unweighted connected graph, and the objective is to
decide whether there is a distance-$1$ solution (see Section~\ref{sec:prel}). As
in~\cite{CyganK14}, we use an integer LP that formulates the problem. If, after
relaxing the integrality constraints, the LP is infeasible, then we know that
there is no distance-$1$ solution, otherwise we round the solution, and obtain
an approximate solution.

In the natural formulation for the capacitated $k$-center, we have opening
variables~$y_u$ for each vertex~$u$, representing the choice of $u$ as a center,
and assignment variables $x_{uv}$ representing that vertex~$v$ is assigned to
center~$u$. In the case of the fault-tolerant $k$-center, for each failure
scenario, that is, for each possible set~$F \subseteq V$ of centers that may
fail, with $|F| \le \alpha$, we must have a different assignment from vertices
to non-faulty centers opened by $y$. One possibility to formulate the
fault-tolerant variant is having different assignment variables for each $F$. To
simplify the formulation, rather than creating a different set of assignment
variables for each failure scenario, we use an equivalent formulation based on
Hall's condition, which is a necessary and sufficient condition for a bipartite
graph to have a perfect matching~\cite{Hall35}. The integer linear program,
denoted by $\ILP$, is the following:
\begin{equation*}
\begin{array}{r@{$\,\,$}l@{\qquad}l}
\sum_{u\in V}y_u  =  &  k   &  \\
| U | \le & \sum_{u \in N_G(U) \setminus F} y_u L_u
          & \forall \; U \subseteq V, \; F \subseteq V : |F| = \alpha\\
  y_u \in & \{ 0 , 1 \}
          & \forall \; u \in V.
\end{array}
\end{equation*}

We remark that $\ILP$ formulates the capacitated $\alpha$-fault-tolerant
$k$-center. The first constraint guarantees that exactly $k$ centers are opened,
and the second set of constraints guarantees that, for each failure scenario,
there is a feasible assignment from clients to opened centers that did not fail.
Indeed, notice that, for a fixed $F$, the existence of such an assignment is
equivalent to the existence of a matching on the bipartite graph formed by
clients and open units of capacity that matches all clients. Hall's result,
together with the second set of constraints of $\ILP$, assures the existence of
such a matching, and thus of such an assignment.

\paragraph{Integrality gap.}

As a first attempt, one can relax $\ILP$ directly. When the integrality
constraints are relaxed, however, the opening fraction on $y$ of a failure
scenario~$F$ of (fractionally opened) centers might be strictly less
than~$\alpha$, that is, $y(F) < \alpha$. Thus the considered constraints are
weaker than desired. Indeed, consider the following example.
Let $C_n$ be a cycle on $n$ vertices, for $n = s^2$ where $s$ is a positive even
integer, and let $G_n$ be the graph obtained from $C_n$ by adding edges between
two vertices at distance at most~$s$ in $C_n$. Note that any pair of antipodes
in $C_n$ are at distance~$s/2$ in~$G_n$. If~$L_u = n$ for every~$u$ in~$G_n$, $k
= s$, and $\alpha = k-1$, then the cost of any solution for this instance
is~$s/2$, as for any set of~$k$ centers in~$G_n$, all but one center might fail.
Now, let $y$ be the vector with $y_u = 1/s$ for every~$u$. We claim that $y$ is
feasible for the relaxation of~$\ILPn$. Indeed, first notice that~$\sum_{u\in V}
y_u = s = k$. Also, since every vertex has~$2s$ neighbors, for any set of
centers~$F$ of size~$\alpha = k-1 = s-1$, the second set of constraints is
satisfied, because the right side is always at least $n$, and the left side is
at most~$n$. So $y$ is feasible, and thus the lower bound obtained from the
relaxation of $\ILPn$ may be arbitrarily small when compared to an optimal
solution, that is, the minimization problem obtained from $\ILPn$ has unbounded
integrality gap.

\subsection{Dealing with the integrality gap}
\label{subsec:clustpreopen}

Suppose that we knew a subset $B$ of the centers of an optimal solution that
might fail. Then we could set $y_u = 1$ for each $u$ in $B$, that is, we decide
opening $u$ before solving the LP. This would avoid the problem in the example
with unbounded integrality gap whenever the failure scenario is $F \subseteq B$,
as in such a case we would have~$y(F) = |F|$.
Since we do not know how to obtain a subset $B$ of centers of an optimal
solution, and a failure scenario $F$ might contain centers not in $B$, we aim at
two more relaxed goals:
\begin{description}
  \item[(G1)] we choose a subset of centers $B$ that are close to distinct
  centers of an optimal solution; and

  \item[(G2)] we assume that only centers in $B$ might fail, and this
  comprises the worst case scenario.
\end{description}

To achieve these goals, we will make use of a standard clustering technique.
Intuitively, a clustering is a partition of the graph so that the elements of
each part are close to some centers in an optimal solution. Locally, the worst
case scenario corresponds to the failure of the highest capacitated centers in a
cluster. The clustering and the selection of pre-opened centers are described
precisely in the following.

\paragraph{Clustering.}

Clustering has been used by several algorithms for the $k$-center problem, for
both the capacitated~\cite{BarilanKP93,KhullerS00} and fault-tolerant
cases~\cite{ChaudhuriGR98,ChechikP15,KhullerPS00,Krumke95}. We use the
construction considered by Khuller and Sussmann~\cite{KhullerS00}. Their
algorithm works by greedily selecting a new vertex $v$ at distance~$3$ from the
set of previous selected centers, and creating a clustering with all not yet
clustered vertices of $N^2(v)$. The relevant result is replicated in next lemma.

\begin{lemma}[\cite{KhullerS00}]\label{lem:monarchtree}
Given a connected graph $G = (V,E)$, one can obtain a set of midpoints
$\Gamma\subseteq V$, and a partition of $V$ into sets $\{C_v\}_{v\in \Gamma}$,
such that
\begin{itemize}
  \item there exists a rooted tree $T$ on $\Gamma$,  with $d_G(u,v) = 3$ for
  every edge~$(u,v)$ of $T$;

  \item $N_G(v) \subseteq C_v$ for every $v$ in $\Gamma$; and

  \item $d_G(u,v) \le 2$ for every $v$ in $\Gamma$ and every $u$ in $C_v$.
\end{itemize}
\end{lemma}

\paragraph{Selecting pre-opened centers.}

We apply Lemma~\ref{lem:monarchtree} and obtain a clustering of $V$. Let $v$
in~$\Gamma$ be a cluster midpoint, and consider any distance-$1$ solution. Since
up to~$\alpha$ centers in this solution may fail, there must be at least $\alpha
+ 1$ centers in~$N(v)$, as otherwise there would be a failure scenario for which
$v$ is not connected. Thus, the elements of $C_v$ are within distance~$3$ from
at least $\alpha+1$ centers in $C_v$. Moreover, since sets $N(v)$ are disjoint
for~$v$ in~$\Gamma$, there are at least $\alpha+1$ centers per cluster in any
distance-$1$ solution.

To achieve (G1), we may select, for each cluster, any subset of up to~$\alpha+1$
vertices in the cluster. To achieve (G2), we reason on the total capacity that
may become unavailable when failure occurs. For each cluster, the largest amount
of capacity that can be discounted in a given scenario does not exceed the
accumulated capacity of the $\alpha$ most capacitated vertices in the cluster.
Thus, we select these vertices as set~$B$.

Formally, for each $v$ in $\Gamma$, let $B_v \subseteq C_v$ be a set of $\alpha$
elements of $C_v$ with largest capacities. This is the set of pre-opened centers
for cluster $C_v$. The set of all pre-opened centers is defined as
\[
B = \cup_{v \in \Gamma} B_v.
\]

\subsection{Modifying the LP formulation}

We pre-open the elements of $B$ by adding to $\ILP$ the constraint $y_u = 1$,
for every $u \in B$. When we establish a partial solution in advance, we may
turn the original linear formulation infeasible, since it is possible that no
distance-$1$ solution opens the elements of~$B$. However, since in any
distance-$1$ solution there are at least $\alpha$ centers in a given cluster,
each center in such a solution is within distance $3$ to a distinct element of
$B$ of non-smaller capacity.
Thus, we can convert a distance-$1$ solution into a distance-$4$ solution by
reassigning clients to elements of $B$, while preserving most of the structure
in the original LP.

\paragraph{Fixing feasibility.}

To obtain a useful LP relaxation, while pre-opening the set $B$ of centers, we
modify the supporting graph $G$. For each cluster $C_v$, we augment $G$ with
edges connecting each client that could be potentially served by centers
in~$C_v$ to each vertex in the set~$B_v$. Precisely, we define the directed
graph $G' = G'(G,\{C_v\}_{v\in\Gamma}) = (V, E')$, where~$E'$ is the set of
arcs~$(u,w)$ such that $\{u,w\} \in E$,  or there exist~$v$ in~$\Gamma$ and~$t$
in~$N(v)$ such that $\{u,t\} \in E$ and $w \in B_v$ (see
Figure~\ref{fig:closure}).
We remark that a directed graph is used, because we want to allow for a
reassignment of a client from an arbitrary center in the cluster to a center
in~$B$, but not the other way around.

\begin{figure}[h]
  \centering
  \scalebox{1.2}{\def\svgwidth{4cm}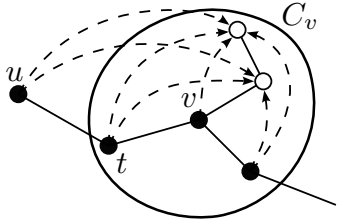}
  \caption{Dashed lines represent arcs added to the graph $G$ to obtain $G'$,
  and solid lines represent duplicated arcs in opposite directions. The white
  vertices represent $B_v$.\label{fig:closure}}
\end{figure}

\paragraph{A new formulation.}

In the new formulation, we consider only scenarios $F \subseteq B$. Thus, in a
feasible solution~$y$, we will have $y(F) = |F|$ for each scenario $F$. Also,
for each cluster midpoint, we want to (fractionally) open at least one
non-faulty center in its neighborhood, for each failure scenario. For the
integer program $\ILP$, this was implicit by the constraints, but when $y$ is
not integral, there might be high capacity centers that satisfy the local demand
with less than one open unit. Therefore, we have an additional constraint for
each cluster midpoint $v$ to ensure that there is one unit of (fractional)
opening in $N(v)$ excluding any opening coming from $B$.
We obtain a new linear program, denoted by $\LP$.
\begin{equation*}
\begin{array}{r@{$\,\,$}l@{\qquad}l}
\sum_{u\in V}y_u  =  &  k   &  \\
| U | \le & \sum_{u \in N_{G'}(U) \setminus F} y_u L_u
          & \forall \; U \subseteq V, \; F \subseteq B : |F| = \alpha\\
    1 \le & \sum_{u \in N_{G}(v) \setminus B} y_u
          & \forall \; v \in \Gamma\\
    y_u = & 1
          & \forall \; u \in B\\
  0 \le  y_u \le & 1
          & \forall \; u \in V.
\end{array}
\end{equation*}

Notice that, contrary to $\ILP$, program $\LP$ depends on the obtained
clustering. The following lemma states that $\LP$ is a ``relaxation'' of $\ILP$,
that is, if $\LP$ is infeasible, then we obtain a certificate that no
distance-$1$ solution for $G$ exists.

\begin{lemma}\label{lem:feasible}
If $\ILP$ is feasible, then $\LP$ is feasible.
\end{lemma}

\begin{proof}
Suppose that $\ILP$ is feasible. Let $y$ be a feasible solution for~$\ILP$,
and let~$R$ be the set of centers corresponding to $y$.

First, we define an injection $\beta$ from $R$ into $R \cup B$ that covers $B$.
Recall that~$\Gamma$ is the set of midpoints. For each $v$ in $\Gamma$, let
$u_1, \dots, u_\alpha$ be the elements of~$B_v$ in non-increasing order of
capacity. Analogously, let $w_1, \dots, w_\alpha, \dots$ be the elements of ${R
\cap N(v)}$ in non-increasing order of capacity (recall that each~$N(v)$ has at
least $\alpha+1$ centers in an optimal solution).  In case of ties, elements
in~$B_v$ should come first in this ordering.  For each~$i$ with $1 \le i \le
\alpha$, we define $\beta(w_i) = u_i$. Finally, for each~$w$ in~$R$
whose~$\beta(w)$ is not yet defined, let $\beta(w) = w$. Notice that, because of
the tie-breaking rule, in this case, $w \not\in B$.  Also, $L_w \le
L_{\beta(w)}$ for every $w$ in $R$, the inverse function $\beta^{-1}$ is
well-defined on the image of $\beta$.

Let $R' = \beta(R)$, and let $y'$ be the characteristic vector of $R'$. We claim
that~$y'$ is a feasible solution for $\LP$. Let $U \subseteq V$ and $F \subseteq
B$ with~$|F| = \alpha$. From the feasibility of $y$ for $\ILP$, and as
$|\beta^{-1}(F)|=|F|=\alpha$, we have
\begin{align*}
  |U| & \textstyle
      \ \le \ \sum_{u \in N_{G}(U) \setminus \beta^{-1}(F)} y_u L_u
      \ = \ \sum_{u \in (N_{G}(U) \setminus \beta^{-1}(F)) \cap R} L_u \\
      &\textstyle
      \ = \ \sum_{u \in (N_{G}(U) \cap R) \setminus \beta^{-1}(F)} L_u
      \ \leq \ \sum_{u \in \beta((N_{G}(U) \cap R) \setminus \beta^{-1}(F))} L_u \\
      &\textstyle
      \ = \ \sum_{u \in \beta(N_{G}(U)\cap R) \setminus F} y'_u L_u
      \ \le \ \sum_{u \in N_{G'}(U) \setminus F} y'_u L_u.
\end{align*}
The verification that the other constraints also hold for $y'$ is
straightforward.
\end{proof}

Though $\LP$ has an exponential number of constraints, the following lemma shows
that it has a polynomial-time separation oracle.

\begin{lemma}\label{lem:sep-LP}
For fixed $\alpha$, there is an algorithm that, in polynomial time, decides
whether a vector~$y$ is feasible for $\LP$. If $y$ is not feasible, the
algorithm also outputs a constraint of $\LP$ that is violated by $y$.
\end{lemma}

\begin{proof}
We concentrate on the second set of constraints, as there are polynomially many
constraints of the other types. Notice that the number of distinct scenarios~$F$
is~$O(|V|^\alpha)$, which is polynomial since $\alpha$ is constant. Fix a
failure scenario $F$ and suppose that we can solve the following problem:
\begin{equation}\label{eq:min-nonconservative}
\min_{U \subseteq V} \sum_{u \in N_{G'}(U) \setminus F} y_u L_u - |U|.
\end{equation}
If this value is non-negative, then all constraints in the second set for this
scenario~$F$ are satisfied, otherwise there is a subset $U^*$ of $V$ for which
the constraint is violated, and we are done. We can rewrite the minimization
problem above as the following integer linear program on binary variables $a_u$
for $u$ in $V$, and $b_u$ for $u$ in $V \setminus F$:
\begin{equation*}
  \begin{array}{rll}
    \min & \multicolumn{2}{l}{\sum_{u \in V \setminus F} b_u (y_u L_u) + \sum_{u \in V} a_u - |V|} \\
    \mbox{s.t.} \ & a_u + b_v \ge 1          & \quad \forall \;\; (u,v) \in G' \\
                  & a_u, \, b_v \in \{0,1\}  & \quad \forall \;\; u \in V, \; v \in V\setminus F.
  \end{array}
\end{equation*}
Variable $a_u$ indicates that $u$ is not in $U$, and variable $b_v$ indicates
that there exists some $u$ in the adjacency list of $v$ that is in $U$ (that is,
$a_u = 0$). The corresponding matrix for this problem is totally unimodular, so
the relaxation has an integral optimal solution, which can be found in
polynomial time. Notice that this problem (excluding the constant $-|V|$ in the
objective function) corresponds to the min-cut formulation for the network flow
problem depicted in Figure~\ref{fig:network-flow}, so it suffices to run any
max-flow min-cut algorithm.
\begin{figure}
  \centering
  \scalebox{1.2}{\def\svgwidth{4cm}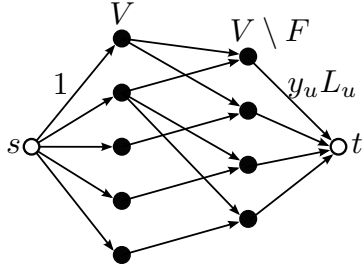}
  \caption{Flow network with source $s$ and target $t$. There is a unit
  capacity arc from $s$ to each $v \in V$, and an uncapacitated arc from $v$
  to each neighbor $u \in N(U)\setminus F$. Moreover, for each $u \in V
  \setminus F$, there is an arc to $t$ with capacity
  $y_uL_u$.\label{fig:network-flow}}
\end{figure}
\end{proof}

\begin{corollary}\label{cor:LPnonconservative}
For fixed $\alpha$, $\LP$ can be solved in polynomial time.
\end{corollary}

If, for an arbitrary $\alpha$, there were a polynomial-time algorithm for
finding a set~$F$ with $|F|=\alpha$ that minimizes the value
of~\eqref{eq:min-nonconservative}, then a stronger version of
Corollary~\ref{cor:LPnonconservative} without the restriction on $\alpha$ being
fixed would hold. Unfortunately, as we show in Section~\ref{sec-complex}, such
algorithm only exists if $\PP = \coNP$.

\subsection{Distance-$r$ transfers}

Given a solution $y$ for $\LP$, the problem of finding $k$ centers to serve all
clients is now reduced to rounding vector $y$ so that exactly $k$ vertices are
integrally open. Since the total fractional opening of $y$ is $k$, one might
consider ``moving'' the fractional opening from one vertex to another so that
the opening of some vertices becomes zero, while the opening of $k$ vertices
become one. This idea motivated the distance-$r$ transfers, introduced by
An~\etal~\cite{AnBCGMS15}, and which we adapt to the fault-tolerant context.

In a distance-$r$ transfer, the fractional opening of vertices are moved to
vertices within distance at most $r$. This guarantees that, after performing
transferring operations, the cost of the solution grows in a controlled way. To
ensure that the capacity constraints are not violated, one might consider only
transferring fractional opening from a low capacitated vertex to higher
capacitated vertices, so that the ``local capacity'' does not decrease. Here, we
use a slightly more general definition than the original one to comprise our
requirements, as we might need to ensure that capacities in certain vertices are
never transferred (that is the case for vertices in $B$), and that transfers
follow certain paths.

\begin{definition}
Let $V$ be a set of vertices,\, $B$ be a subset of $V$, \, $H = (W, E)$ be a
graph with $W \subseteq V$, \, $L : W \rightarrow \Rnn$ be a capacity function
on $W$, and $y \in \Rnn^{|V|}$. A vector $y'$ in $\Rnn^{|V|}$ is an
\emph{$H$-restricted distance-$r$ transfer of $y$} if
\begin{enumerate}[(a)]
  \item $\sum_{v \in W} y'_v = \sum_{v \in W} y_v$;

  \item $\sum_{v \in N^r_H(U) \setminus B} L_v y'_v \ge \sum_{v \in U \setminus
  B} L_v y_v$ for every $U \subseteq W$; and

  \item $y'_v = y_v$ for every $v \in (V \setminus W) \cup B$.
\end{enumerate}
\end{definition}

If $y'$ is the characteristic vector of a set $R \subseteq V$, we will say that
$R$ is an \emph{integral $H$-restricted distance-$r$ transfer of $y$}. If $H =
G$, then we simply say that $y'$ is a \emph{distance-$r$ transfer of $y$}.

An~\etal~\cite{AnBCGMS15} reduced the rounding of an arbitrary graph to the case
in which the graph is a tree that satisfies certain properties. They showed that
such trees have integral distance-$2$ transfers. This is formalized in the
following.

\begin{lemma}[\cite{AnBCGMS15}]\label{lem:Anetal}
Let $T = (W, E)$ be a tree with $W \subseteq V$, and $y$ in $[0,1]^{|V|}$ be a
vector such that~$\sum_{u \in W} y_u$ is an integer, and $y_v = 1$ for every
internal node $v$ of~$T$ in~$W$. One can find in polynomial time an integral
($T$-restricted) distance-$2$ transfer of~$y$.
\end{lemma}

For a given solution $y$ for $\LP$ and any failure scenario $F \subseteq B$, the
LP implicitly defines an assignment of clients to non-failed (fractionally
opened) centers at distance~$1$ in~$G'$. Suppose some portion of the opening
$y_v$ of~$v \notin B$ is transfered to some other vertex $v'$ at distance~$r$ in
$G$. If a client $u$ is initially served by $v$, then the assignment can be
transfered to $v'$ as well, so that~$u$ will be (fractionally) assigned to
centers at distance at most $r+1$ in~$G$, as~$(u, v) \in E[G]$. If client~$u$
was initially served by some $v \in B$, then this assignment may be left
unchanged, as no opening of~$v$ is transfered; in this case, however, we might
have $(u, v) \not\in E[G]$, and so edge~$(u,v)$ in~$G'$ may correspond to a path
of length~$4$ in~$G$. The worst case of the obtained assignment happens when the
distance is the maximum between $r+1$ and~$4$.

\subsection{The algorithm}

Our algorithm consists of two parts. In the first, we round a fractional
solution $y$ of $\LP$, and obtain a set $R$ of $k$ centers. In the second part,
for each failure scenario $F \subseteq R$ with $|F| \leq \alpha$, we have to
obtain an assignment from $V$ to $R \setminus F$.

\paragraph{Rounding.}

Since we have pre-opened centers, we round only the residual set of vertices $V
\setminus B$. This phase is based on the algorithm of An~\etal~\cite{AnBCGMS15}
for the capacitated (non-fault-tolerant) $k$-center.
The main difference is that we do not allow transfers from or to vertices in the
set~$B$. The algorithm reduces the problem of rounding a general graph to the
problem of rounding tree instances. There are three consecutive transfers. In
the first step, we concentrate one unit of opening on one auxiliary vertex that
is added at the same location as the cluster midpoint.
In the second step, we create a tree instance using the auxiliary vertices as
internal nodes, and obtain an integral transfer using Lemma~\ref{lem:Anetal}. In
the last step, the opening of auxiliary vertices is transferred back to vertices
of the original graph. A detailed description is presented in the following:

\begin{enumerate}[Step 1.]
  \item For each cluster $C_v$, choose an element $m_v$ in the neighborhood of
  the midpoint~$v$ that is not pre-opened, and has the largest capacity, that
  is, $m_v = \argmax_{u \in N_G(v) \setminus B} L_u$. Create an auxiliary
  vertex $a_v$ at the same location as $v$ (add an edge to $a_v$ from each
  element of $N(v)$ as in Figure~\ref{fig:aux}), with capacity
  $L_{a_v}=L_{m_v}$, and initial opening~$y_{a_v} = 0$. Next, aggregate one
  unit of opening to~$a_v$ by transferring fractional openings from
  $N_G(v)\setminus B$ to~$a_v$. This can be done as $\sum_{u \in N_{G}(v)
  \setminus B} y_u \ge 1$. The transfer proceeds as follows: for each $u$ in
  $N_G(v) \setminus B$, decrease~$y_u$, while increasing $y_{a_v}$,
  until~$y_u$ becomes~$0$. The process is interrupted once~$y_{a_v}$
  reaches~$1$. The result is a distance-$1$ transfer~$\y{1}$. The first vertex
  to have its fractional opening transferred is $m_v$, so that, at the end of
  this step, $\y{1}_{m_v} = 0$.

  \item Obtain a tree $T$ from the clustering tree by replacing each
  midpoint~$v$ with $a_v$ for every~$v$ in $\Gamma$. Next, for each cluster
  $C_v$, select every vertex~$u$ in $C_v$ such that $0 < \y{1}_u < 1$ and add
  a leaf corresponding to~$u$, connected to~$a_v$. Finally, apply
  Lemma~\ref{lem:Anetal}, and obtain an integral $T$-restricted distance-$2$
  transfer $\y{2}$ (starting with $\y{1}$). Notice that $d_G(w_1, w_2) = 3$
  for each edge $(w_1, w_2)$ of $T$ if both $w_1$ and $w_2$ are internal
  nodes; and $d_G(w_1, w_2) \le 2$ if either~$w_1$ or~$w_2$ is a leaf. Hence,
  $\y{2}$ can be interpreted as a distance-$(2 \cdot 3)$ transfer of~$\y{1}$
  (on the graph~$G$).

  \item For each cluster $C_v$, transfer the opening of the auxiliary vertex
  $a_v$ back to the original vertex $m_v$. This is possible since $\y{2}_{m_v}
  = 0$. Obtain a final integral distance-$1$ transfer~$\y{3}$. Open the set of
  vertices $R \subseteq V$ that corresponds to the characteristic
  vector~$\y{3}$.
\end{enumerate}

\begin{figure}
  \centering
  \scalebox{1.2}{\def\svgwidth{2.8cm}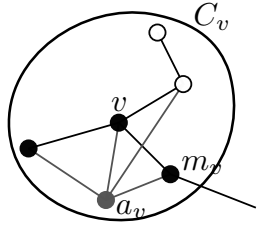}
  \caption{Auxiliary vertex $a_v$ is adjacent to each vertex in $N(v)$. The
  white vertices represent set $B_v$.\label{fig:aux}}
\end{figure}

\paragraph{Assignment.}

After opening centers $R$, up to $\alpha$ failures might occur. Our algorithm
must provide a valid assignment for each failure scenario $F \subseteq R$. We
will consider two cases, depending on whether $F \subseteq B$.

First, we examine the case that $F$ is a subset of $B$. In this case, $\LP$
assures the existence of an assignment from $V$ to a set of fractionally opened
centers that does not intersect $F$. Since the rounding algorithm obtains an
integral $G$-restricted distance-$8$ transfer (by adding up the three
consecutive transfers), this will lead to a distance-$9$ solution that does not
assign to any element of $F$.

For the case that $F$ is not a subset of $B$, we may not rely on the existence
of a fractional assignment obtained from the LP. Instead, we will show that, for
each~$F$, there exists a corresponding $F' \subseteq B$, and that a distance-$9$
solution for failure scenario~$F'$ can be transformed into a distance-$10$
solution for failure scenario~$F$. Indeed, we will show that each element $u$
assigned to a center $v \in F \setminus F'$ in the former solution may be
reassigned to a distinct element $v' \in F' \setminus F$ in the latter solution,
such that $v$ and $v'$ are in the same cluster, and $L_{v'} \ge L_v$.

A naive analysis of the preceding strategy would yield a $13$-approximation, as
the distance between $v$ and $v'$ might be $4$, and thus $d(u,v') \le d(u,v) +
d(v,v') \le 9 + 4$.
To obtain a more refined analysis, we will bound the distance between $u$ and
the midpoint associated to $v$. More precisely, denote by $\delta(v)$ the
midpoint of the cluster that contains $v$.
We obtain the following lemma.

\begin{lemma}\label{lemma:assign}
Consider $F \subseteq B$ with $|F| = \alpha$ and let $R$ be the integral
transfer obtained from $y$ by the rounding algorithm above. One can find, in
polynomial time, an assignment $\phi : V \rightarrow R\setminus F$ such that
$d_G(u, \phi(u)) \le 9$ and $d_G(u, \delta(\phi(u))) \le 8$ for each~$u$ in~$V$.
\end{lemma}

\begin{proof}
\newcommand{\ic}[1]{\textit{ic}^{\scriptscriptstyle(#1)}\!}

Let $\bar{V} = V \cup \{a_v : v \in \Gamma\}$ be the union of the vertices of
$G$ and the auxiliary vertices, and $\bar{G}$ be the graph obtained after we add
the auxiliary vertices to $G$. Fix a subset $U \subseteq V$.

Recall that the rounding algorithm considers an initial feasible solution ${y =
\y{0}}$, and obtains consecutive transfers $\y{1}, \y{2}, \y{3}$. In the
following, for each transfer $\y{i}$, for $0 \le i \le 2$, and each $U\subseteq
V$, we will consider a set $X = \X{i}(U)$ of vertices, excluding faulty
elements, whose total installed capacity exceeds~$|U|$. That is, we want to
obtain $X$ such that the value $\ic{i}(X) := \sum_{u \in X \setminus F} \y{i}_u
L_u \ge |U|$.
Initially, before any transfer is performed, we have that $\y{0} = y$ and, by
the constraints of~$\LP$, we have that $|U| \le \sum_{u \in N_{G'}(U) \setminus
F} y_u L_u$, so we set $\X{0}(U) = N_{G'}(U)$.

In the first step, we have a distance-$1$ transfer. Notice that ${N_{G'}(U)
\setminus B = N_{G}(U) \setminus B}$. Also, recall that the transfer is
restricted to vertices in $\bar{V}\setminus B$. We obtain
\begin{align*}
|U| \le \ic{0}(\X{0}(U))
  & = \ic{0}((N_{G'}(U) \cap B) \cup (N_{G}(U) \setminus B)) \\
  & \le \ic{1}((N_{G'}(U) \cap B) \cup (N^2_{G}(U) \setminus B)).
\end{align*}
Hence we set $\X{1}(U) = (N_{G'}(U) \cap B) \cup (N^2_{G}(U) \setminus B)$.

In the second step, we have an integral $T$-restricted distance-$2$ transfer
of~$\y{1}$. Once again, since $T$ does not include vertices of $B$, we obtain
\begin{align*}
\ic{1}(\X{1}(U))
  & = \ic{1}((N_{G'}(U) \cap B) \cup (N^2_{G}(U) \setminus B)) \\
  & \le \ic{2}( (N_{G'}(U) \cap B)\cup N^2_T(N^2_{G}(U) \setminus B)).
\end{align*}
We set $\X{2}(U) = (N_{G'}(U) \cap B)\cup N^2_T(N^2_{G}(U) \setminus B)$.

\medskip

Let $\bar{R}  \subseteq \bar{V}$ be the set corresponding to vector $\y{2}$.
First consider a bipartite graph ${H = (V \cup \bar{R}, D)}$ with $\{u, v\}$ in
$D$ if $v \in \X{2}(\{u\}) \setminus F$.
Then modify $H$ by including ${L_v-1}$ additional copies of each vertex $v$ in
$\bar{R}$. Notice that now, for each $U \subseteq V$, we have~$|N_H(U)| =
|\bigcup_{u \in U} N_H(\{u\})| = \ic{2}(\bigcup_{u \in U} \X{2}(\{u\})) =
\ic{2}(\X{2}(U)) \ge |U|$. This is exactly Hall's condition for the existence of
a matching in $H$ covering $V$. We obtain such a matching in polynomial time,
and obtain a corresponding assignment $\phi' : V \rightarrow \bar{R}$.

Now, for every vertex $u$ in $V$, we show that the distance from $u$ to
$\delta(\phi'(u))$ is bounded by~$8$. Recall that~$\phi'(u) \in N_H(\{u\}) =
\X{2}(\{u\})$. We have two cases. First, suppose that $\phi'(u) \in
N_{G'}(\{u\}) \cap B$. Since we have that $d_G(u, \phi'(u)) \le 4$, by the
construction of $G'$, we obtain that $d_G(u, \delta(\phi'(u))) \le 6$, and we
are done. Now, assume that $\phi'(u) \in N^2_T(N^2_{G}(\{u\}) \setminus B)$. In
this case, there must be some $v$ in $N^2_{G}(\{u\})\setminus B$ and a shortest
path $\rho$ connecting $v$ to~$\phi'(u)$ in~$T$.
We consider two possibilities. If the length of $\rho$ is $1$, then $\rho =
(v,\phi'(u))$, and we deduce that $d_G(u, \delta(\phi'(u))) \le {d_G(u, v) +
d_G(v, \phi'(u)) +  d_G(\phi'(u), \delta(\phi'(u)))}  \le  {2 + 3 + 2  = 7}$. If
the length of $\rho$ is $2$, then there exists~$w$ such that $\rho = (v, w,
\phi'(u))$, and we get that  $d_G(u, \phi'(u)) \le {d_G(u, v) + d_G(v, w) +
d_G(w,\phi'(u))} \le {2 + 3 + 3 = 8}$. If $\phi'(u)$ is an internal node of $T$,
then $\delta(\phi'(u)) = \phi'(u)$, and thus $d_G(u, \delta(\phi'(u))) \le 8$.
Otherwise, $w$ must be an internal node and $\phi'(u)$ a leaf of $w$. Hence
$\delta(\phi'(u)) = w$, and therefore $d_G(u, \delta(\phi'(u))) \le {d_G(u,
\phi'(u)) \le 8}$.

A similar analysis also allows us to deduce that $d_G(u, \phi'(u)) \le 8$. To
obtain a final assignment $\phi : V \rightarrow R$, we reassign each vertex $u$
assigned to an auxiliary vertex $a_v$, to the vertex~$m_v$, that is, for each
$u$ in $V$, if $\phi'(u) = a_v$ for some $v$ in~$\Gamma$, then set $\phi(u) =
m_v$, otherwise set~$\phi(u) = \phi'(u)$.
\end{proof}

Now we may obtain the approximation factor.

\begin{theorem}\label{thm:approx}
There exists a $10$-approximation for the capacitated $\alpha$-fault-tolerant
$k$-center with fixed $\alpha$.
\end{theorem}

\begin{proof}
Consider a failure scenario $F \subseteq R$ with $|F| = \alpha$. For each
cluster $C_v$, let $F_v$ be the set of centers that failed in cluster $C_v$.
Also, let $F'_v$ be the set of the $|F_v|$ most capacitated centers in~$B_v$,
and $F' = \bigcup_{v \in \Gamma} F'_v$. We use Lemma~\ref{lemma:assign}, and
obtain an assignment $\phi' : V \rightarrow R \setminus F'$.
Now, for each $v$ in $\Gamma$, obtain an ordering $\{u_1,\ldots, u_t\}$ of the
vertices in~$F'_v \setminus F_v$, and an ordering $\{v_1,\ldots, v_t\}$ of the
vertices in~$F_v \setminus F'_v$.
For each $w$ that is assigned to~$v_i$, for some $1 \le i\le t$, reassign it to
$u_i$, that is, for every $w$ such that~$\phi'(w) = v_i$, set $\phi(w) = u_i$.
Notice that this leads to a valid assignment $\phi$, since $L_{u_i} \ge L_{v_i}$
for every $1 \le i \le t$. Also, we notice that since $u_i$ and $v_i$ are in the
same cluster, $d_G(\delta(v_i), u_i) \le 2$, and thus $d_G(w, \phi(w)) \le
d_G(w, \delta(\phi'(w))) + d_G(\delta(\phi'(w)), u_i)  \le 8 + 2  = 10$.
\end{proof}

\section{$\{0,L\}$-Capacitated fault-tolerant $k$-center}
\label{sec:aprox0L}

For a given $L$, the $\{0,L\}$-capacitated fault-tolerant $k$-center is the
particular version of the capacitated fault-tolerant $k$-center in which every
vertex has capacity either zero or $L$. Vertices with capacity~$0$ are called
\emph{$0$-vertices} and vertices with capacity $L$ are called
\emph{$L$-vertices}. For a given set $A$ of vertices, we denote by~$A^L$ the set
containing all~$L$-vertices of~$A$.

\subsection{LP-formulation}

We give a rounding algorithm for the $\{0,L\}$-capacitated case. As in
Section~\ref{sec:noncons}, we formulate the problem using $\ILP$. In this case,
however, we may rewrite the program such that only $L$-vertices appear in the
summation, and all coefficients are equal, that is, $\ILP$ can be written as:

\begin{equation*}
\begin{array}{r@{$\,\,$}l@{\qquad}l}
\sum_{u\in V}y_u  =  &  k   &  \\
| U | \le & \sum_{u \in (N_G(U))^L \setminus F} y_u L
          & \forall \; U \subseteq V, \; F \subseteq V : |F| \le \alpha\\
  y_u \in & \{ 0 , 1 \}
          & \forall \; u \in V.
\end{array}
\end{equation*}

Notice that the second line in the program above can be simplified. The key
observation is that, in the worst case, the total failed capacity is always the
constant $\alpha L$. Indeed, consider a feasible integer solution $y$ and a
fixed subset $U \subseteq V$ such that $U \ne \emptyset$, and let ${H = \{ u \in
(N_G(U))^L : y_u =  1\}}$. We have $|H| > \alpha$, since otherwise we would get
$
{|U| \le \sum_{u \in (N_G(U))^L \setminus H} y_u L =}
{\sum_{u \in (N_G(U))^L\setminus H} 0 \cdot L = 0},
$
that is a contradiction since $U$ is not empty.
Let $F'$ be any subset of $H$ with $|F'| = \alpha$. From
the inequality constraint in $\ILP$ for $F = F'$,
we obtain
\begin{align*}
\textstyle
|U| \le \sum_{u \in (N_G(U))^L \setminus F'} y_u L
    &\textstyle=    \sum_{u \in (N_G(U))^L} y_u L
                   -\sum_{u \in (N_G(U))^L \cap F'} 1\cdot L\\
    &\textstyle=    \sum_{u \in (N_G(U))^L} y_u L
                   -\alpha L.
\end{align*}

Therefore, the following linear program, that is denoted by $\LPUU$, is a
relaxation of $\ILP$.
\begin{equation*}
\begin{array}{r@{$\,\,$}l@{\qquad}l}
\sum_{u\in V}y_u  =  &  k   &  \\
| U | \le & \sum_{u \in (N_{G}(U))^L} y_u L - \alpha L
          & \forall \; U \subseteq V, \; U \neq \emptyset \\
1     \le & \sum_{u \in (N_{G}(v))^L} y_u
          & \forall \; v \in V \\
  0 \le  y_u \le & 1
          & \forall \; u \in V.
\end{array}
\end{equation*}

In contrast to~$\LP$, this program can be separated even if~$\alpha$ is part of
the instance. The difference is that, in this formulation, the failure scenarios
need not be enumerated. Given a candidate solution~$y$, we can compute the
minimum value of $\sum_{u \in N^L(U)} y_u L - |U|$ over all sets $U$, and check
whether this value is at least $\alpha L$. This can be done in polynomial time
using a max-flow min-cut algorithm with arguments very similar to those in the
proof of Lemma~\ref{lem:sep-LP}. This means that we can separate $\LPUU$ in
polynomial time, which implies the following lemma.

\begin{lemma}
$\LPUU$ can be solved in polynomial time even if~$\alpha$ is part of the input.
\end{lemma}

\subsection{Rounding}

For the non-fault-tolerant $\{0,L\}$-capacitated $k$-center,
An~\etal~\cite{AnBCGMS15} perform an additional preprocessing of the input graph
to obtain a clustering with stronger properties. This way, they derive an
integral distance-$5$ transfer. Namely, before the preprocessing described in
Section~\ref{sec:prel}, which produces an unweighted connected graph~$G=(V,E)$,
they remove any edge connecting two $0$-vertices. We apply their rounding
algorithm to the solution obtained for $\LPUU$, obtaining the following result.

\begin{lemma}\label{lemma-antransf}
Suppose $G$ is a connected graph such that each vertex is either a $0$-vertex or
an $L$-vertex, no two adjacent vertices are $0$-vertices, and $y$ is a feasible
solution for $\LPUU$. Then there is a polynomial-time algorithm that produces an
integral distance-$5$ transfer $y'$ of $y$.
\end{lemma}

Now we obtain a $6$-approximation the the $\{0,L\}$-capacitated case.

\begin{theorem}\label{thm:approx-uniform}
There exists a $6$-approximation for the $\{0,L\}$-capacitated
$\alpha$-fault-tolerant $k$-center (with $\alpha$ as part of the input).
\end{theorem}

\begin{proof}

Let $y$ be an optimal solution for $\LPUU$, and $y'$ be an integral distance-$5$
transfer of~$y$ obtained by the algorithm of Lemma~\ref{lemma-antransf}. Also,
let $R$ be the set of centers corresponding to the characteristic vector $y'$.
We proceed as in the proof of Lemma~\ref{lemma:assign}. Consider a subset $U
\subseteq V$.
Let $X(U) =\{ v : y_v > 0 \mbox{ and } v \in  (N_{G}(U))^L\}$, and let
$Y(U)\subseteq R$ be the set of integrally opened centers to which we have
transfered fractional opening from $X(U)$. By the constraints of $\LPUU$, and
the fact that $y'$ is an integral transfer, we get
\[
|U| + \alpha L \le \sum_{u \in X(U)} y_u L \le \sum_{u \in Y(U)} y'_u L = |Y(U)|
L.
\]

Now consider a failure scenario $F \subseteq V$ with $|F| = \alpha$. We can
create a bipartite graph (as in Lemma~\ref{lemma:assign}) that connects each
vertex $u \in V$ to vertices $Y(\{u\}) \setminus F \subseteq R$. Using Hall's
condition, we obtain an assignment $\phi : V \rightarrow R\setminus F$ that
respects the capacities. Since $y'$ is a distance-$5$ transfer, we know that
$Y(\{u\}) \subseteq N^6(\{u\})$ for every~$u$, and thus $d(u, \phi(u)) \le
6$.
\end{proof}

\section{The $k$-supplier}
\label{sec-supplier}

In this section, we consider the \emph{$k$-supplier problem}, which is a common
variant of the $k$-center. In this problem, disjoint sets of clients
$\mathcal{C}$ and facility locations $\mathcal{F}$ are given, and one must
select $k$ facilities to serve each of the clients. In the capacitated
fault-tolerant version, each client must be assigned to a facility, even at the
failure of up to $\alpha$ facilities, and the assignment is such that no
facility $u$ is assigned more than $L_u$ clients.
In the following, we show that our algorithms naturally extend to this
generalization, for both the conservative and non-conservative variants.
Table~\ref{tab:summary2} summarizes the obtained factors.

\begin{table}[h]
\renewcommand{\arraystretch}{1.1}
\centering
\begin{tabular}{|@{\ }c@{\ }|@{\ }c@{\ }|@{\ }c@{\ }|@{\ }c@{\ }|}
\hline
\textbf{Version} & \textbf{Capacities} & \textbf{Value of $\alpha$}
& \textbf{Factor} \\
\hline
conservative     & uniform   & given in the input & $7$          \\
conservative     & arbitrary & fixed              & $11+8\alpha$ \\
non-conservative & uniform   & given in the input & $7$          \\
non-conservative & arbitrary & fixed              & $13$         \\\hline
\end{tabular}
\caption{\label{tab:summary2} Summary of the obtained approximation factors for
the $k$-supplier problem.}
\end{table}

As in the case of the $k$-center, we reduce the problem to the case of a
unweighted connected graph $G$ and the objective is to obtain a distance-$1$
solution. For the $k$-supplier, however, we consider only edges between
$\mathcal{C}$ and $\mathcal{F}$, that is, the obtained graph is bipartite. This
implies that distances in $G$ between pairs of clients or between pairs of
facilities are even.

\paragraph{The non-conservative capacitated fault-tolerant $k$-supplier.}

We first consider the case that capacities are non-uniform. A slightly different
formulation from $\ILP$ is used: the main difference is that we only have
variables $y_u$ for elements $u$ of $\mathcal{F}$, and we only consider
constraints corresponding to subsets of clients $U \subseteq \mathcal{C}$ and
failure scenarios $F \subseteq \mathcal{F}$.

By adapting the example of Section~\ref{subsec-iniLP}, the obtained formulation
also has unbounded integrality gap, and thus we consider a relaxation based on a
modified graph that depends on a clustering. In this step, rather than using the
clustering by Khuller and Sussmann~\cite{KhullerS00}, we greedily pick clients
whose distance to previously picked elements is exactly $4$. This set of
elements $\Gamma$ (midpoints) induces a clustering of $\mathcal{F}$, and a
corresponding tree of midpoints such that any adjacent midpoints in the tree are
at distance $4$, and every facility is associated to a midpoint at distance at
most $3$.

As in the case of the $k$-center, we select a set $B_v$ of $\alpha$ facilities
of largest capacity in each cluster centered at $v \in \Gamma$, and construct a
graph $G'$ by adding arcs from any client at distance $2$ from a midpoint $v$ to
each facility of $B_v$. Let $B$ be the union of all $B_v$, for $v \in \Gamma$.
The obtained LP relaxation is:
\begin{equation*}
\begin{array}{r@{$\,\,$}l@{\qquad}l}
\sum_{u\in \mathcal{F}}y_u  =  &  k   &  \\
| U | \le & \sum_{u \in N_{G'}(U) \setminus F} y_u L_u
          & \forall \; U \subseteq \mathcal{C}, \; F \subseteq B : |F| = \alpha\\
    1 \le & \sum_{u \in N_{G}(v) \setminus B} y_u
          & \forall \; v \in \Gamma\\
    y_u = & 1
          & \forall \; u \in B\\
  0 \le  y_u \le & 1
          & \forall \; u \in V.
\end{array}
\end{equation*}

As done in~\cite{AnBCGMS15}, a rounding algorithm similar to that for the
$k$-center can obtain an integral distance-$10$ transfer of a solution for the
previous linear program (the only difference is that a distance-$2$ transfer on
the tree of midpoints is now interpreted as a distance-$2\cdot 4$ transfer on
the original graph). This transfer implies that, for a failure scenario
${F\subseteq B}$, one may obtain an assignment $\phi$ such that $d(u, \phi(u))
\le 11$ for every $u \in \mathcal{C}$. Moreover, by using the same reasoning as
in the proof of Lemma~\ref{lemma:assign}, one may show that $d(u,
\delta(\phi(u))) \le 10$, where $\delta(\phi(u))$ is the midpoint associated
with $\phi(u)$. Therefore, a distance-$11$ assignment for a failure scenario
$F'\subseteq B$ can be transformed into a distance-$13$ assignment for a general
failure scenario $F\subseteq \mathcal{F}$.

For the uniformly capacitated case, we can also obtain a simplified relaxation
as in Section~\ref{sec:aprox0L}. It is straightforward to adapt the rounding
algorithm for the $\{0,L\}$-capacitated $k$-center by An~\etal~\cite{AnBCGMS15},
and obtain an integral distance-$6$ transfer for the solution for this
relaxation. The reason that the algorithm obtains a distance-$6$ transfer for
the $k$-supplier, rather than a distance-$5$ transfer, is that cluster midpoints
are at distance~$4$ in an instance of the $k$-supplier, whist midpoints are at
distance $3$ in an instance of the $k$-center.
Now, repeating the arguments in the proof of Theorem~\ref{thm:approx-uniform},
we obtain a $7$-approximation for the uniformly capacitated fault-tolerant
$k$-supplier.

\paragraph{The conservative capacitated fault-tolerant $k$-supplier.}

First, we revisit the notion of independent sets for the $k$-supplier. A set of
\emph{facilities} $W$ is $(\alpha,\ell)$-independent if each connected component
of $G^\ell[W]$ contains at most $\alpha$ vertices. Also, a set of \emph{clients}
$A$ is $8$-independent if $d(u,v) \ge 8$ for every $u,v \in A$ (notice that, in
this bipartite setting, requiring that a set of clients is $7$-independent is
the same as requiring that it is $8$-independent). With these adapted
definitions, one may obtain versions of Lemmas~\ref{lem:alpha-independent}
and~\ref{lem:7independent} with analogous statements.

For the uniformly capacitated case, we use Algorithm~\ref{alg:conservative-0L},
but with an $8$-independent set $A$, and assuming that \alg is a
$\beta$-approximation for the capacitated $k$-supplier problem. Notice that
since $A$ is maximal, for every client $u$, there is a client $v \in A$, such
that $d(u, v) \le 6$. Now, by repeating Theorem~\ref{thm:conservative0L}, we
obtain that this algorithm has approximation factor $\max\{7, \beta\}$. We use
the algorithm by An~\etal for the uniformly capacitated case (without failures),
for which, as stated above, $\beta = 7$, and obtain a $7$-approximation.

For the non-uniformly capacitated case, we use Algorithm~\ref{alg:conservative}.
However, when augmenting the set of backup facilities $B$ with a set of
facilities $U$ (Line~\ref{alg:augment-B}), rather than excluding elements in
$N^6(U) \cap B$, we exclude the elements in $N^8(U) \cap B$.
Recall that, in the $k$-center problem, we obtain a $7$-independent set $A
\subseteq B$ by selecting an element in each connected component of $G^6[B]$
(see the proof of Lemma~\ref{lem:preopenBgeral}).
In the the case of the $k$-supplier problem, to obtain an $8$-independent set
$A$ of clients, we must choose from the neighborhood of the set $B$ of backup
facilities (and not directly from~$B$).
Thus, for each connected component $C_i$ of $G^8[B]$, we choose a facility $b_i
\in C_i$, and a neighbor $a_i \in N(b_i) \in \mathcal{C}$. Notice that, for any
pair $b_i, b_j$, $d(b_i, b_j) > 8$, thus $d(b_i, b_j) \ge 10$, and hence $d(a_i,
a_j) \ge 8$.
Therefore, the set $A$ of all $a_i$'s is an $8$-independent set. The rest of the
proof remains unchanged, except that we replace~$6$ by~$8$, obtaining a factor
${\beta + 8 \alpha}$.
The best known approximation for the capacitated $k$-supplier has factor $\beta
=  11$~\cite{AnBCGMS15}.

\section{Complexity results}
\label{sec-complex}

The following theorem shows that the subproblem solved by
Algorithm~\ref{alg:conservative} is \coNP-complete when $\alpha$ is part of the
input.

\begin{theorem}\label{thm:coNP1}
The problem of, given a graph $H=(V_H,E_H)$, a number $L_u$ for each $u \in
V_H$, a set $B \subseteq V_H$, and a number $\alpha$, deciding whether ${L(U)
\le L(B \cap N(U))}$ for every $U \subseteq V_H$ with $|U| = \alpha$ is
\coNP-complete.
\end{theorem}

\begin{proof}
  This problem is in \coNP, because, for an instance $(H,L,B,\alpha)$ whose
  answer to the problem is no, that is, a {\sc no} instance, one can present
  as a {\sc no} certificate a set $U \subseteq V_H$ such that $|U|=\alpha$ and
  $L(U) > L(B \cap N(U))$.

  The clique problem is known to be \NP-complete~\cite{Karp72} and consists
  in, given a graph $G$ and a positive integer~$k$, to decide whether there
  exists a clique in $G$ with at least~$k$ vertices.  We present a reduction
  from the clique problem to our problem so that an instance $(G,k)$ of the
  clique problem is a {\sc yes} instance if and only if the corresponding
  instance $(H,L,B,\alpha)$ for our problem is a {\sc no} instance.

  Let $(G,k)$ be an instance of the clique problem with $G=(V,E)$. The main
  part of the graph $H$ consists of the bipartite graph with bipartition
  $\{V,E\}$ and a vertex $v$ in $V$ adjacent to an edge $e$ in~$E$ if $v$ is
  an end of $e$ in $G$.  Besides this, graph~$H$ has two disjoint cliques
  on~$k+1$ vertices, say~$C_V$ and $C_E$. A vertex in~$C_V$, say~$s$, is
  adjacent to each vertex in~$V$ while a vertex in $C_E$, say~$t$, is adjacent
  to each edge in $E$.  This finishes the description of graph $H$.  See
  Figure~\ref{fig:reduction1} for an example. The capacity function~$L$ is
  defined as follows. For each~$e$ in~$E$, let $L(e)=1$; for each $v$ in~$V$,
  let $L(v)$ be the degree of~$v$ in $G$, denoted as~$d_v$; for each $u$ in
  $C_V$, let $L(u)=\binom{k}{2}-1$ and, for each $u$ in~$C_E$, let $L(u)=|E|$.
  Finally, let $B = E \cup C_V \cup C_E$ and $\alpha = k$.  This concludes the
  description of the instance of our problem, which can be obtained from
  $(G,k)$ in time polynomial in the size of $(G,k)$.  Next we show
  that~$(G,k)$ is a {\sc yes} instance for the clique problem if and only
  if~$(H,L,B,\alpha)$ is a {\sc no} instance of our problem.

  \begin{figure}[b]
    \centering
    \scalebox{1.2}{\def\svgwidth{9.35cm}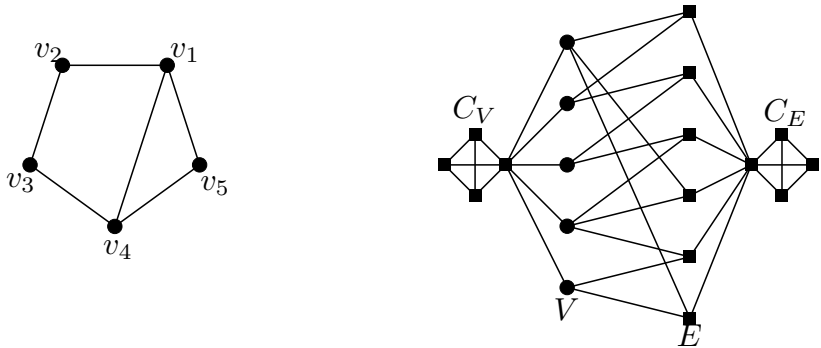}
    \caption{The graph on the right is the graph $H$ in the instance of our
    problem corresponding to the graph $G$ on the left for $k=3$. Squared
    vertices indicate the set $B$.\label{fig:reduction1}}
  \end{figure}

  First let us prove that, if there exists a clique $S$ of size $k$ in $G$,
  then ${L(S) > L(B \cap N(S))}$, that is, the answer of our problem for the
  instance $(H,L,B,\alpha)$ is {\sc no}.  Indeed, $B \cap N(S)$ consists of
  the special vertex $s$ in~$C_V$ and the edges incident to $S$ in $G$, so
  $L(B \cap N(S)) = \binom{k}{2} - 1 + \ell$, where~$\ell$ is the number of
  edges in $G$ incident to $S$.  The value of $L(S)$ is $\sum_{v \in S} d_v$,
  which is exactly the number of edges incident to $S$ plus the number of
  edges in $G$ with both ends in $S$, that is, the edges in the graph $G[S]$
  induced by $S$.  As the number of edges in $G[S]$ is exactly $\binom{k}{2}$
  because $S$ is a clique on $k$ vertices, $L(S) = \binom{k}{2} + \ell > L(B
  \cap N(S))$, as we wished.

  Second we prove that, if ${L(U) > L(B \cap N(U))}$ for a set $U$ of $k$
  vertices of~$H$, then there is a clique with $k$ vertices in $G$.  We start
  by arguing that ${L(U \cap (C_V \cup C_E))} \leq {L((B \cap N(U) \cap (C_V
  \cup C_E)) \setminus \{s,t\})}$. If~${U \cap C_V \neq \emptyset}$, then~$B
  \cap N(U) \supseteq C_V$. Moreover, $U \neq C_V$ since $C_V$ has $k+1$
  vertices and~$U$ has $k$ vertices. This means that ${L(U \cap C_V)} \leq
  {L((B \cap N(U) \cap C_V) \setminus \{s\})}$. Similarly, if $U \cap C_E \neq
  \emptyset$, then $B \cap N(U) \supseteq C_E$. Again $U \neq C_E$, so we have
  that~${L(U \cap C_E)} \leq {L((B \cap N(U) \cap C_E) \setminus \{t\})}$,
  completing the proof of the claimed inequality. Now note that $L(U \cap E)
  \leq |E| = L(t)$.  On the other hand, let~$S = U \cap V$ and note that $L(S)
  = \sum_{v \in S} d_v$, which is exactly the number of edges incident to $S$
  plus the number of edges in the graph $G[S]$.  If $S$ is not a clique on~$k$
  vertices, then $L(S) \leq L(E \cap N(U)) + L(s)$ and, joining everything, we
  deduce that $L(U) \leq L(B \cap N(U))$, a contradiction.  So $S$ must be a
  clique on~$k$ vertices in~$G$.
\end{proof}

Analogously, the following theorem shows that the separation problem for program
$\LP$ is $\coNP$-hard when $\alpha$ is part of the input. Thus, to achieve a
constant approximation for capacitated fault-tolerant $k$-center with general
capacities and $\alpha$ as part of input, one needs a different strategy. Notice
that this is equivalent to the problem of deciding whether a subset of $S
\subseteq V$ is a distance-$1$ solution for the capacitated fault-tolerant
$k$-center.

\begin{theorem}\label{thm:coNP2}
The problem of, given a graph $H=(V_H,E_H)$, a number $L_u$ for each $u \in
V_H$, and a number $\alpha$, deciding whether $L(N(U) \setminus F) \ge |U|$ for
every $U \subseteq V_H$ and $F \subseteq V_H$ with $|F|=\alpha$ is
\coNP-complete.
\end{theorem}

\begin{proof}
  The proof is similar to that of Theorem~\ref{thm:coNP1}. It is easy to see
  that the problem is in \coNP as, for a {\sc no} instance of the problem, one
  can present as a certificate subsets~$U$ and $F$ of $V_H$ such that
  $|F|=\alpha$ and $L(N(U) \setminus F) < |U|$.

  Consider again the \NP-complete clique problem: given a graph $G$ and a
  positive integer~$k$, decide whether there exists a clique in $G$ with at
  least~$k$ vertices.  Next we present a reduction from the clique problem to
  our problem so that an instance~$(G,k)$ of the clique problem is a {\sc yes}
  instance if and only if the corresponding instance $(H,L,\alpha)$ for our
  problem is a {\sc no} instance.

  Let $(G,k)$ be an instance of the clique problem, where $G=(V,E)$. The main
  part of the graph $H$ consists of the bipartite graph with $V$ as one side
  and $E$ as the other side of the bipartition.  A vertex $v$ in $V$ is
  adjacent to an edge $e$ in~$E$ if~$v$ is an end of $e$ in $G$.  Besides
  this, $H$ has two disjoint cliques, say, $C_V$ on~$k+1$ vertices and $C_E =
  A_E \cup B_E$ on~$2k+1$ vertices, with $|A_E|=k$ and $|B_E|=k+1$. Every
  vertex in $C_V$ is adjacent to each vertex in~$V$ and every vertex in $A_E$
  is adjacent to each edge in~$E$. This finishes the description of graph~$H$.
  See Figure~\ref{fig:reduction2} for an example.  As for~$L$, for each~$e$
  in~$E$, let $L(e) = 0$; for each $v$ in~$V$, let $L(v) = d_v$, where $d_v$
  is the degree of~$v$ in~$G$; for each $u$ in $C_V \cup B_E$, let $L(u) =
  |V_H|$; denoting by $a_1,\ldots,a_k$ the vertices in~$A_E$, let $L(a_i) = i$
  for $i=1,\ldots,k-1$ and $L(a_k) = k-1$. Finally, let~$\alpha = k$,
  concluding the description of the instance of our problem, which can be
  obtained from $(G,k)$ in time polynomial in the size of $(G,k)$.  Next we
  show that~$(G,k)$ is a {\sc yes} instance for the clique problem if and only
  if~$(H,L,\alpha)$ is a {\sc no} instance of our problem.

  \begin{figure}
    \centering
    \scalebox{1.2}{\def\svgwidth{9.435cm}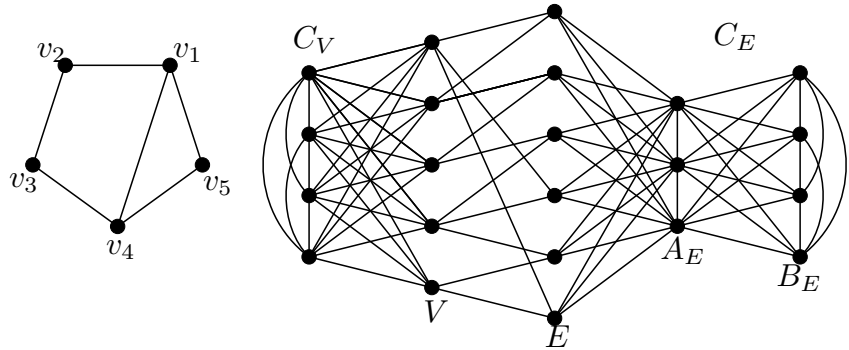}
    \caption{The graph on the right is the graph $H$ in the instance of our
    problem corresponding to the graph on the left for
    $k=3$.\label{fig:reduction2}}
  \end{figure}

  First, suppose that there exists a clique $S$ of size $k$ in $G$.  Let $U$
  be the edges in $G[S]$ and let $F = S$. Note that $|F| = |S| = k = \alpha$
  and $|U| = \binom{k}{2}$, because $S$ is a clique on $k$ vertices.  Thus
  $N(U) \setminus F = (S \cup A_E) \setminus F = A_E$, and $L(N(U) \setminus
  F) = L(A_E) = \binom{k}{2} - 1 < |U|$. Hence the answer of our problem for
  the instance~$(H,L,\alpha)$ is {\sc no}.

  Second, suppose that there are subsets $U$ and $F$ of $V_H$ such that $|F| =
  \alpha$ and $L(N(U) \setminus F) < |U|$.  Observe that $U \cap (V \cup C_V)
  = \emptyset$, otherwise $N(U) \supseteq C_V$ and~$C_V \setminus F \neq
  \emptyset$ because~$F$ has $k$ vertices and $C_V$ has $k+1$ vertices.  But
  this would mean that~$L(N(U) \setminus F) \ge |V_H| \ge |U|$, a
  contradiction.  Similarly $U \cap C_E = \emptyset$, otherwise $N(U)
  \supseteq B_E$ and, as $B_E$ has $k+1$ vertices, $L(N(U) \setminus F) \ge
  |V_H| \ge |U|$, a contradiction. So we know that~$U \subseteq E$.  Now let
  $S = N(U) \cap V$ and let $\ell = |S \cap F|$. Thus $N(U)$ has at least
  $\ell$ vertices in $A_E \setminus F$, and then $L(N(U) \setminus F) \ge
  L(\{a_1,\ldots,a_\ell\}) + L(S \setminus F)$. Notice that
  $L(\{a_1,\ldots,a_\ell\}) = \binom{\ell}{2}$, if~$\ell < k$, and
  $L(\{a_1,\ldots,a_k\}) =  \binom{k}{2}-1$. On the one hand, the number of
  edges in~$G[S \cap F]$ is at most~$\binom{\ell}{2}$, because $|S \cap F| =
  \ell$.  On the other hand, $L(S \setminus F) = \sum_{u \in S \setminus F}
  d_u$, which is the number of edges incident to~$S \setminus F$ plus the
  number of edges in~$G[S \setminus F]$. Thus $L(N(U) \setminus F) \ge |U|$
  unless $\ell = k$ and~$G[S \cap F]$ is a clique on~$k$ vertices in~$G$.
  Hence, as $L(N(U) \setminus F) < |U|$, there is a clique in~$G$ on $k$
  vertices.
\end{proof}

\bibliographystyle{abbrv}
\bibliography{kcenterarxiv}

\end{document}